\documentclass[11pt]{article}
\usepackage{mathrsfs}
\usepackage{CJK}
\usepackage{bbm}
\usepackage{graphicx}
\usepackage{amsmath,amssymb}
\usepackage{amsfonts}
\usepackage{amsthm}
\usepackage{verbatim}
\usepackage{color}

\usepackage[numbers,sort&compress]{natbib}

\newtheorem{Definition}{Definition}[section]
\newtheorem{Theorem}{Theorem}[section]
\newtheorem{Remark}{Remark}[section]
\newtheorem{Lemma}{Lemma}[section]

\newtheorem{Proposition}{Proposition}[section]

\begin{document}

\title{Recovery of signals under the high order RIP condition via prior support information }
\author{Wengu~Chen and Yaling~Li
\thanks{W. Chen is with Institute of Applied Physics and Computational Mathematics,
Beijing, 100088, China, e-mail: chenwg@iapcm.ac.cn.}
\thanks{Y. Li is with Graduate School, China Academy of Engineering Physics,
Beijing, 100088, China, e-mail:leeyaling@126.com.}
\thanks{This work was supported by the NSF of China (Nos.11271050, 11371183)
and Beijing Center for Mathematics and Information Interdisciplinary
Sciences (BCMIIS).} }

\maketitle

\begin{abstract}
In this paper we study the recovery conditions of weighted $l_{1}$
minimization for signal reconstruction from incomplete linear
measurements when partial prior support information is available. We
obtain that a high order RIP condition can guarantee stable and
robust recovery of signals in bounded $l_{2}$ and Dantzig selector
noise settings. Meanwhile, we not only prove that the sufficient
recovery condition of weighted $l_{1}$ minimization method is weaker
than that of standard $l_{1}$ minimization method, but also prove
that weighted $l_{1}$ minimization method provides better upper
bounds on the reconstruction error in terms of the measurement noise
and the compressibility of the signal, provided that the accuracy of
prior support estimate is at least $50\%$. Furthermore, the
condition is proved sharp.
\end{abstract}


\section{Introduction}
{Compressed}  sensing is a new type of sampling theory that admits
that high dimensional sparse signals can be reconstructed through
fewer measurements than their ambient dimension. The central goal in
compressed sensing is to recover a signal $x\in\mathbb{R}^{N}$ based
on $A$ and $y$ from the following model:
\begin{align}\label{m1}
y=Ax+z
\end{align}
where sensing matrix $A\in \mathbb{R}^{n\times N}$ with $n\ll N$,
i.e., using very few measurements, $y\in \mathbb{R}^{n}$ is a vector
of measurements, and $z \in \mathbb{R}^{n}$ is the measurement
error. In past decade, compressed sensing has triggered considerable
research in a number of fields including applied mathematics,
statistics, electrical engineering, seismology and signal
processing. Compressed sensing is especially promising in
applications where taking measurements is costly, e.g.,
hyperspectral imaging \cite{DF}, as well as in applications where
the ambient dimension of the signal is very large, e.g., medical
image reconstruction \cite{LSLDP}, DNA microarrays \cite{PVMH},
radar system \cite{BS,HS,ZZZ}.

To reconstruct the signal $x$ from (\ref{m1}), Cand\`es and Tao
\cite{CT} proposed the following constrained $l_{1}$ minimization
method:
\begin{align}\label{f1}
  \underset{x\in \mathbb{R}^{N}}{\rm minimize}\quad\|x\|_{1} \ \ \ {\rm subject\quad to}\ \ \
  \|y-Ax\|_2\leq\epsilon.
\end{align}
It is well known that $l_{1}$ minimization
 is a convex relaxation of $l_{0}$ minimization and is polynomial-time
solvable.
 And it has been shown that $l_{1}$ minimization is an effective way to recover sparse signals
 in many settings \cite{CZ2,CWX1,CWX,CZ1,ML,CZ,CRT,CXZ}.
 Cai and Zhang \cite{CZ,CZ2} established sharp restricted isometry conditions to achieve the exact and stable
 recovery of signals
 in both noiseless and noisy cases via $l_{1}$ minimization method.

Note that compressed sensing is a nonadaptive data acquisition
technique and $l_{1}$ minimization method (\ref{f1}) is itself
nonadaptive because no prior information on the signal $x$ is used.
In practical examples, however, the estimate of the support of the
signal or of its largest coefficients may be possible to be drawn.
For example, support estimation of the previous time instant
 may be applied to recover time sequences of sparse signals iteratively.
 Incorporating prior information is very useful for recovering signals from compressive measurements.
Thus, the following weighted $l_{1}$ minimization method  which
incorporates partial support information of the signals has been
introduced to replace standard $l_{1}$ minimization
\begin{align}\label{f3}
 \underset{x\in \mathbb{R}^{N}}{\rm minimize}\quad\|x\|_{1,\mathrm{w}} \ \ \
 {\rm subject\quad to}\ \ \  \|y-Ax\|_2\leq\epsilon,
 \end{align}
where $\mathrm{w} \in [0, 1]^{N}$ and
$\|x\|_{1,\mathrm{w}}=\sum\limits_{i}\mathrm{w}_{i}|x_{i}|$.
Reconstructing compressively sampled signals with partially known
support has been previously studied in the literature; see
\cite{BMP,VL,LV,J,KXAH,LV1,FMSY}. Borries, Miosso and Potes in
\cite{BMP}, Khajehnejad $et~al.$ in \cite{KXAH}, and Vaswani and Lu
in \cite{VL} introduced the problem of signal recovery with
partially known support independently. The works by Borries $et~al.$
in \cite{BMP}, Vaswani and Lu in \cite{VL,LV,VL1} and Jacques in
\cite{J} incorporated known support information using weighted
$l_{1}$ minimization approach with zero weights on the known
support, namely, given a support estimate
$\widetilde{T}\subset\{1,2,\ldots, N\}$ of unknown signal $x$,
setting $\mathrm{w}_{i}=0$ whenever $i\in \widetilde{T}$ and
$\mathrm{w}_{i}=1$ otherwise, and derived sufficient recovery
conditions. Friedlander $et~al.$ in \cite{FMSY} extended weighted
$l_{1}$ minimization approach to nonzero weights. They allow the
weights $\mathrm{w}_{i}=\omega\in [0, 1]$ if $i\in \widetilde{T}$.
Since Friedlander $et~al.$ incorporated the prior support
information and consider the accuracy of the support estimate, they
derived the stable and robust recovery guarantees for weighted
$l_{1}$ minimization which generalize the results of Cand\`es,
Romberg and Tao in \cite{CRT} stated below. Friedlander $et~al.$
\cite{FMSY} pointed out that once at least $50\%$ of the support
information is accurate,
  the weighted $l_{1}$ minimization method (\ref{f3}) can stably and robustly recover any signals under weaker sufficient
  conditions than the analogous conditions for standard  $l_{1}$ minimization method (\ref{f1}).
  In addition, the weighted $l_{1}$ minimization method (\ref{f3}) gives better upper bounds on the reconstruction error.
Furthermore, they also pointed out sufficient conditions are weaker
than those of \cite{VL} when $\omega=0$.

To recover sparse signals via constrained $l_{1}$ minimization,
Cand\`es and Tao \cite{CT} introduced the notion of Restricted
Isometry Property (RIP), which is one of the most commonly used
frameworks for compressive sensing. The definition of RIP is as
follows.
\begin{Definition}
Let $A\in \mathbb{R}^{n\times N}$ be a matrix and $1\leq k \leq N$
is an integer. The restricted isometry constant (RIC) $\delta_{k}$
of order $k$ is defined as the smallest nonnegative
 constant that
satisfies
$$(1-\delta_{k})\|x\|_{2}^{2}\leq\|Ax\|_{2}^{2}\leq(1+\delta_{k})\|x\|_{2}^{2},$$
for all $k-$sparse vectors $x\in\mathbb{R}^{N}.$ A vector $x\in
\mathbb{R}^{N}$ is $k-$sparse if $|supp(x)|\leq k$, where
$supp(x)=\{i: x_{i}\neq 0\}$ is the support of $x$. When $k$ is not
an integer, we define $\delta_{k}$ as $\delta_{\lceil k \rceil}$,
where $\lceil k \rceil$ denotes the smallest integer strictly bigger
than $k$.
\end{Definition}

Cand\`es, Romberg and Tao \cite{CRT} showed that the condition
$\delta_{ak}+a\delta_{(a+1)k}<a-1$ with $a\in \frac 1k\mathbb{Z}$
and $a>1$
 is sufficient for stable and robust recovery of all signals using $l_{1}$ minimization method (\ref{f1}).
  Cai and Zhang \cite{CZ} improved the result of Cand\`es, Romberg and Tao \cite{CRT} and proved that the condition $\delta_{tk}<\sqrt{\frac{t-1}{t}}$ with $t\geq 4/3$
can guarantee the exact recovery of all $k-$sparse signals in the
noiseless case and stable recovery of  approximately sparse signals
in the noise case by $l_{1}$ minimization method (\ref{f1}).
Furthermore, Cai and Zhang proved that for any $\epsilon>0$,
$\delta_{tk}<\sqrt{\frac{t-1}{t}}+\epsilon$ fails to ensure the
exact reconstruction of all $k-$sparse signals and stable
reconstruction of approximately sparse signals for large $k$.

In \cite{CZ}, Cai and Zhang use the following $l_{1}$ minimization
 \begin{align}\label{f4}
  \underset{x\in \mathbb{R}^{N}}{\rm minimize}\quad\|x\|_{1} \ \ \ {\rm subject\quad to}\ \ \
  \|y-Ax\|_2\in \mathcal{B},
\end{align}
where $\mathcal{B}$ is a bounded set determined by the noise
structure, and $\mathcal{B}$ is especially taken to be $\{0\}$ in
the noiseless case. They consider two types of noise settings
\begin{align}\label{b1}
 \mathcal{B}^{l_{2}}(\varepsilon)=\{z: \|z\|_{2}\leq\varepsilon\}
\end{align}
and
\begin{align}\label{b2}
 \mathcal{B}^{DS}(\varepsilon)=\{z:
 \|A^{T}z\|_{\infty}\leq\varepsilon\}.
\end{align}
 In this paper, we adopt the corresponding weighted $l_{1}$ minimization
 method:
\begin{align}\label{f2}
 \underset{x\in \mathbb{R}^{N}}{\rm minimize}\quad&\|x\|_{1,\mathrm{w}}\ \ \
 {\rm subject\quad to}
 \ \ \ y-Ax\in \mathcal{B}  \notag \\
 {\rm with }\ \ \ \mathrm{w}_{i}&=\left\{\begin{array}{cc}
                     1, & i\in \widetilde{T}^{c} \\
                     \omega, &  i\in \widetilde{T}.
                   \end{array}
                   \right.
 \end{align}
where $0\leq \omega \leq 1 $ and $\widetilde{T}\subset\{1,2,\ldots,
N\}$ a given support estimate of unknown signal $x$. $\mathcal{B}$
is also a bounded set determined by the noise settings (\ref{b1})
and (\ref{b2}). Our goal is to generalize the results of Cai and
Zhang \cite{CZ} via the weighted $l_{1}$ minimization method
(\ref{f2}). We establish the high order RIP condition for the stable
and robust recovery of signals with partially known support
information from (\ref{m1}). We also show that the recovery by
weighted $l_{1}$ minimization method (\ref{f2}) is stable and robust
under weaker sufficient conditions compared to the standard $l_{1}$
minimization method (\ref{f4}) when we have the partial support
information with accuracy better than $50\%$.

The rest of the paper is organized as follows. In Section \ref{2},
we will introduce some notations and some basic lemmas that will be
used. The main results  are given in Section \ref{3}, and the proofs
of our main results are presented in Section \ref{4}.

\section{Preliminaries}\label{2}
Let us begin with basic notations. For arbitrary
$x\in\mathbb{R}^{N}$, let $x_{k}$ be its best $k-$term
approximation. $x_{\max(k)}$ is defined as $x$ with all but the
largest $k$ entries in absolute value set to zero, and
$x_{-\max(k)}=x-x_{\max(k)}$. Let $T_{0}$ be the support of $x_{k}$,
i.e., $T_{0}={\rm supp}(x_{k})$, with $T_{0}\subseteq
\{1,\ldots,N\}$ and $|T_{0}|\leq k$. Let $\widetilde{T}\subseteq
\{1, \ldots, N\}$ be the support estimate of $x$ with
$|\widetilde{T}|=\rho k$, where $\rho \geq 0$ represents the ratio
of the size of the estimated support to the size of the actual
support of $x_k$ (or the support of $x$ if $x$ is $k-$ sparse).
Denote $\widetilde{T}_{\alpha}=T_{0}\cap \widetilde{T}$ and
$\widetilde{T}_{\beta}=T_{0}^{c}\cap \widetilde{T}$ with
$|\widetilde{T}_{\alpha}|=\alpha |\widetilde{T}|=\alpha\rho k$ and
$|\widetilde{T}_{\beta}|=\beta |\widetilde{T}|=\beta\rho k$, where
$\alpha$ denotes the ratio of the number of indices in $T_0$ that
were accurately estimated in $\widetilde{T}$ to the size of
$\widetilde{T}$ and $\alpha+\beta=1$. For arbitrary nonnegative
number $\xi$, we denote by $[[\xi]]$ an integer satisfying $\xi\leq
[[\xi]] <\xi+1.$

Cai and Zhang developed a new elementary technique which is a key
technical tool for the proof of the main result (see Theorem 3.1).
It states that any point in a polytope can be represented as a
convex combination of sparse vectors (\cite{CZ}, Lemma 1.1). Another
 key technical tool for our proof was Lemma \ref{l2} introduced by Cai and Zhang  (\cite{CZ1}, Lemma 5.3).
 The specific contents are presented in Lemmas \ref{l1} and \ref{l2}, respectively.
\begin{Lemma}[\cite{CZ}, Lemma 1.1]\label{l1}
For a positive number $\alpha$ and a positive integer $k$, define
the polytope $T(\alpha, k)\subset\mathbb{R}^{d}$ by
$$T(\alpha, k)=\{v\in\mathbb{R}^{d}: \|v\|_{\infty}\leq \alpha, \|v\|_{1}\leq k\alpha\}.$$
For any $v\in \mathbb{R}^{d}$, define the set of sparse vectors
$U(\alpha, k, v)\subset\mathbb{R}^{d}$ by
\begin{align*}
U(\alpha, k, v)=\{&u\in\mathbb{R}^{d}: supp(u)\subseteq supp(v), \|u\|_{0}\leq k, \\
&\|u\|_{1}=\|v\|_{1}, \|u\|_{\infty}\leq\alpha\},\end{align*} where
$\|u\|_{0}=|supp(u)|$. Then any $v\in T(\alpha, k)$ can be expressed
as
$$v=\sum\limits_{i=1}^{N}\lambda_{i}u_{i},$$
where $u_{i}\in U(\alpha, k, v)$ and $0\leq \lambda_{i}\leq 1,
\sum\limits_{i=1}^{N}\lambda_{i}=1.$
\end{Lemma}

\begin{Lemma}[\cite{CZ1}, Lemma 5.3]\label{l2}
Assume $m\geq k$, $a_{1}\geq a_{2}\geq\cdots\geq a_{m}\geq 0$,
$\sum\limits_{i=1}^{k}a_{i}\geq \sum\limits_{i=k+1}^{m}a_{i},$ then
for all $\alpha\geq 1$,
$$\sum\limits_{j=k+1}^{m}a_{j}^{\alpha}\leq \sum\limits_{i=1}^{k}a_{i}^{\alpha}.$$
More generally, assume $a_{1}\geq a_{2}\geq\cdots\geq a_{m}\geq 0$,
$\lambda\geq 0$ and $\sum\limits_{i=1}^{k}a_{i}+\lambda\geq
\sum\limits_{i=k+1}^{m}a_{i},$ then for all $\alpha\geq 1$,
$$\sum\limits_{j=k+1}^{m}a_{j}^{\alpha}\leq k\Big(\sqrt[\alpha]{\frac{\sum_{i=1}^{k}a_{i}^{\alpha}}{k}}+\frac{\lambda}{k}\Big)^{\alpha}.$$
\end{Lemma}

As we mentioned in the introduction, Cai and Zhang \cite{CZ}
provided the sharp sufficient condition to
 recover sparse signals and approximately sparse signals via $l_{1}$ minimization (\ref{f4}).
 Their main result can be stated as below.
\begin{Theorem}[\cite{CZ}, Theorem 2.1]\label{t1}
Let $y=Ax+z$ with $\|z\|_{2}\leq \varepsilon$ and
$\widehat{x}^{l_{2}}$ is the minimizer of {\rm(\ref{f4})} with
$\mathcal{B}=\mathcal{B}^{l_{2}}(\eta)=\{z:\|z\|_{2}\leq \eta\}$ for
some $\eta \geq \varepsilon$. If
\begin{align}\label{g1}
 \delta_{tk}<\sqrt{\frac{t-1}{t}}
\end{align}
for some $t\geq 4/3$, then
\begin{align}\label{g2}
  \|\widehat{x}^{l_{2}}-x\|_{2}\leq C_{0}(\varepsilon+\eta)
  +C_{1}\frac{2\|x_{-\max{(k)}}\|_{1}}{\sqrt{k}},
\end{align}
where
\begin{align}\label{g21}
  C_{0}&=\frac{\sqrt{2t(t-1)(1+\delta_{tk})}}{t(\sqrt{(t-1)/t}-\delta_{tk})}, \notag\\
   C_{1}&=\frac{\sqrt{2}\delta_{tk}+\sqrt{t(\sqrt{(t-1)/t}-\delta_{tk})\delta_{tk}}}{t(\sqrt{(t-1)/t}-\delta_{tk})}+1.
\end{align}
Let $y=Ax+z$ with $\|A^{T}z\|_{\infty}\leq \varepsilon$ and
$\widehat{x}^{DS}$ is the minimizer of {\rm(\ref{f4})} with
$\mathcal{B}=\mathcal{B}^{DS}(\eta)=\{z:\|A^{T}z\|_{\infty}\leq
\eta\}$ for some $\eta \geq \varepsilon$. If
$\delta_{tk}<\sqrt{\frac{t-1}{t}}$
 for some $t\geq 4/3$, then
\begin{align}\label{g3}
  \|\widehat{x}^{DS}-x\|_{2}\leq  C'_{0}(\varepsilon+\eta)
  + C'_{1}\frac{2\|x_{-\max{(k)}}\|_{1}}{\sqrt{k}},
\end{align}
where
\begin{align}\label{g22}
  C'_{0}=\frac{\sqrt{2t^{2}(t-1)k}}{t(\sqrt{(t-1)/t}-\delta_{tk})},\ \ \ C'_{1}=C_{1}.
\end{align}
\end{Theorem}
Note that Theorem \ref{t1} always holds for $ t>1$, and the
condition $t\geq 4/3$ ensures that (\ref{g1}) is sharp.

 Friedlander et al. \cite{FMSY} used the prior support
information to recover any signals
 by weighted $l_{1}$ minimization (\ref{f2}). The following theorem was showed in \cite{FMSY}.
\begin{Theorem}[\cite{FMSY}, Theorem 3]\label{t2}
Let $x\in \mathbb{R}^{N}$ be an arbitrary signal and $y=Ax+z$ with
$\|z\|_{2}\leq \varepsilon$. Define $x_{k}$ be its best $k-$term
approximation with ${\rm supp\{x_{k}\}}=T_{0}$. Let
$\widetilde{T}\subseteq \{1, \ldots, N\}$ be an arbitrary set and
define $\rho \geq 0$ and $0\leq \alpha \leq 1$ such that
$|\widetilde{T}|=\rho k$ and $|\widetilde{T}\cap T_{0}|=\alpha \rho
k.$ Suppose that there exists an $a\in \frac{1}{k}\mathbb{Z}$ with
$a\geq(1-\alpha)\rho$ and $a>1$. If the measurement matrix $A$ has
RIP satisfying
\begin{align}\label{g4}
 \delta_{ak}+\frac{a}{\gamma^{2}}\delta_{(a+1)k}<\frac{a}{\gamma^{2}}-1,
\end{align}
where $\gamma=\omega+(1-\omega)\sqrt{1+\rho-2\alpha\rho}$ for some
given $0\leq \omega \leq 1$. Then the solution $\widehat{x}$ to
{\rm(\ref{f2})} with {\rm(\ref{b1})} obeys
\begin{align}\label{g5}
  \|\widehat{x}-x\|_{2} \leq C''_{0}(2\varepsilon)
  +C''_{1}\frac{2\left(\omega\|x-x_{k}\|_{1}+(1-\omega)\|x_{\widetilde{T}^{c}\cap T_{0}^{c}}\|_{1}\right)}{\sqrt{k}},
\end{align}
where
\begin{equation}\label{g6}
  \begin{split}
      C''_{0}&=\frac{1+\frac{\gamma}{\sqrt{a}}}
  {\sqrt{1-\delta_{(a+1)k}}-\frac{\gamma}{\sqrt{a}}\sqrt{1+\delta_{ak}}},  \\
C''_{1}&=\frac{a^{-1/2}\left(\sqrt{1-\delta_{(a+1)k}}+\sqrt{1+\delta_{ak}}\right)}
 {\sqrt{1-\delta_{(a+1)k}}-\frac{\gamma}{\sqrt{a}}\sqrt{1+\delta_{ak}}}.
   \end{split}
\end{equation}
\end{Theorem}

\begin{Remark}[\cite{FMSY}, Remarks 3.3 and 3.4]\label{r1}
If $A$ satisfies
\begin{align}\label{g7}
  \delta_{(a+1)k}<\delta_{a}^{\omega}:=\frac{a-\gamma^{2}}{a+\gamma^{2}},
\end{align}
 where
$\gamma=\omega+(1-\omega)\sqrt{1+\rho-2\alpha\rho}$, then Theorem
\ref{t2} holds with same constants. If
\begin{align}\label{g8}
   \delta_{2k}<\left(\sqrt{2}\gamma+1\right)^{-1},
\end{align}
then weighted $l_{1}$ minimization (\ref{f2}) with {\rm(\ref{b1})}
can stably and robustly recover  the original signal.
\end{Remark}
\section{Main results}\label{3}
\begin{Theorem}\label{t3}
Suppose that $x\in\mathbb{R}^{N}$ be an arbitrary signal and $x_{k}$
be its best $k-$term approximation supported on $T_{0}\subseteq \{1,
\ldots, N\}$ with $|T_{0}|\leq k$. Let $\widetilde{T}\subseteq \{1,
\ldots, N\}$ be an arbitrary set and denote $\rho \geq 0$ and $0\leq
\alpha \leq 1$ such that $|\widetilde{T}|=\rho k$ and
$|\widetilde{T}\cap T_{0}|=\alpha \rho k.$
 Let $y=Ax+z$ with $\|z\|_{2}\leq\varepsilon$ and $\widehat{x}^{l_{2}}$ is the minimizer of (\ref{f2}) with (\ref{b1}). If the measurement matrix $A$ satisfies  RIP with
\begin{align}\label{g23}
  \delta_{tk}<\delta_{t}^{\omega}:=\sqrt{\frac{t-d}{t-d+\gamma^{2}}}
\end{align}
 for $t> d$, where $\gamma=\omega+(1-\omega)\sqrt{1+\rho-2\alpha\rho}$ and
 \begin{align*}
   d=\left\{
 \begin{array}{cc}
   1, & \omega=1 \\
   1-\alpha\rho+a,  & 0\leq\omega<1
 \end{array}
 \right.
 \end{align*}
 with $a=\max{\{\alpha, \beta\}}\rho$.
Then
\begin{align}\label{g9}
\|\widehat{x}^{l_{2}}-x\|_{2} \leq D_{0}(2\varepsilon)
+D_{1}\frac{2\left(\omega\|x_{T_{0}^{c}}\|_{1}+(1-\omega)\|x_{\widetilde{T}^{c}\cap
T_{0}^{c}}\|_{1}\right)}{\sqrt{k}},
 \end{align}
where
\begin{align}\label{g24}
D_{0}=& \frac{\sqrt{2(t-d)(t-d+\gamma^{2})(1+\delta_{tk})}}{(t-d+\gamma^{2})(\sqrt{\frac{t-d}{t-d+\gamma^{2}}}-\delta_{tk})},\notag\\
D_{1}
=&\frac{\sqrt{2}\delta_{tk}\gamma+\sqrt{(t-d+\gamma^{2})(\sqrt{\frac{t-d}{t-d+\gamma^{2}}}-\delta_{tk})\delta_{tk}}}
{(t-d+\gamma^{2})(\sqrt{\frac{t-d}{t-d+\gamma^{2}}}-\delta_{tk})}
+\frac{1}{\sqrt{d}}.
\end{align}

  Let $y=Ax+z$ with $\|A^{T}z\|_{\infty}\leq\varepsilon$. Assume that $\widehat{x}^{DS}$ is the minimizer of (\ref{f2}) with (\ref{b2}) and the matrix $A$ satisfies  RIP (\ref{g23}).
Then
 \begin{align}\label{g10}
\|\widehat{x}^{DS}-x\|_{2} &\leq  D'_{0}(2\varepsilon)  + D'_{1}
\frac{2\left(\omega\|x_{T_{0}^{c}}\|_{1}+(1-\omega)\|x_{\widetilde{T}^{c}\cap
T_{0}^{c}}\|_{1}\right)}{\sqrt{k}},
\end{align}
where
\begin{equation}\label{g25}
  \begin{split}
    D'_{0}=\frac{\sqrt{2(t-d)(t-d+\gamma^{2})[[tk]]}}{(t-d+\gamma^{2})(\sqrt{\frac{t-d}{t-d+\gamma^{2}}}-\delta_{tk})},\ \
    D'_{1}=D_{1}.
   \end{split}
\end{equation}

\end{Theorem}

\begin{Remark}
In Theorem \ref{t3}, every signal $x\in\mathbb{R}^{N}$ can be stably
and robustly recovered. And if $\mathcal{B}=\{0\}$ and $x$ is a
$k-$sparse vector, then Theorem \ref{t3} ensures exact recovery of
the signal $x$.
\end{Remark}

For Gaussian noise case,  the above results on the bounded noise
case can be directly applied to yield the corresponding results by
using the same argument as in \cite{CXZ,CWX1}. The concrete content
is stated as follows.
\begin{Remark}
Let $x\in\mathbb{R}^{N}$ be an arbitrary signal and $x_{k}$ be its
best $k-$term approximation supported on $T_{0}\subseteq \{1,
\ldots, N\}$ with $|T_{0}|\leq k$. Let $\widetilde{T}\subseteq \{1,
\ldots, N\}$ be an arbitrary set and define $\rho \geq 0$ and $0\leq
\alpha \leq 1$ such that $|\widetilde{T}|=\rho k$ and
$|\widetilde{T}\cap T_{0}|=\alpha \rho k.$ Assume that $z\sim
{\mathcal{N}}_{n}(0, \sigma^{2}I)$ in (\ref{m1}) and
$\delta_{tk}<\sqrt{\frac{t-d}{t-d+\gamma^{2}}}$ for $t> d$. Let
$\mathcal{B}^{l_{2}}=\{z: \|z\|_{2}\leq\sigma\sqrt{n+2\sqrt{n\log
n}}\}$ and $\mathcal{B}^{DS}=\{z:
\|A^{T}z\|_{\infty}\leq\sigma\sqrt{2\log N}\}$.
$\widehat{x}^{l_{2}}$ and $\widehat{x}^{DS}$ are the minimizer of
(\ref{f2}) with $\mathcal{B}^{l_{2}}$ and $\mathcal{B}^{DS}$,
respectively. Then, with probability at least $1-1/n$,
\begin{align*}
  \|\widehat{x}^{l_{2}}-x\|_{2}
&\leq D_{0}(2\sigma\sqrt{n+2\sqrt{n\log n}}) +D_{1}
\frac{2\left(\omega\|x_{T_{0}^{c}}\|_{1}+(1-\omega)\|x_{\widetilde{T}^{c}\cap
T_{0}^{c}}\|_{1}\right)}{\sqrt{k}},
\end{align*}
and
\begin{align*}
  \|\widehat{x}^{DS}-x\|_{2}
&\leq D'_{0}(2\sigma\sqrt{2\log N}) +D'_{1}
\frac{2\left(\omega\|x_{T_{0}^{c}}\|_{1}+(1-\omega)\|x_{\widetilde{T}^{c}\cap
T_{0}^{c}}\|_{1}\right)}{\sqrt{k}},
\end{align*}
 with probability at least $1-1/\sqrt{\pi\log N}$.
 Here
 $ d=\left\{
 \begin{array}{cc}
   1, & \omega=1 \\
   1-\alpha\rho+a, & 0\leq\omega<1
 \end{array}
 \right.$
 with $a=\max{\{\alpha, \beta\}}\rho$, and
 $\gamma=\omega+(1-\omega)\sqrt{1+\rho-2\alpha\rho}$.
\end{Remark}

\begin{Theorem}\label{t4}
Let $d=1$ and $t\geq
1+\frac{\left(1-\sqrt{1-\gamma^{2}}\right)^{2}}{\gamma^{2}+2\left(1-\sqrt{1-\gamma^{2}}\right)}$.
For any $\varepsilon>0$ and $k\geq \frac{6}{\varepsilon}$. Then
there exists a sensing matrix $A\in \mathbb{R}^{n\times N}$ with
$\delta_{tk}<\sqrt{\frac{t-d}{t-d+\gamma^{2}}}+\varepsilon$
 and some $k-$sparse signal $x_{0}$ such that
 \begin{description}
   \item[{\rm (1)}]In the noiseless case, i.e., $y=Ax_{0}$,
   the weighted $l_{1}$ minimization (\ref{f2}) can not exactly recover the $k-$sparse signal $x_{0}$, i.e., $\widehat{x}\neq x_{0}$,
   where $\widehat{x}$ is the solution to (\ref{f2}).
   \item[{\rm (2)}] In the noise case, i.e., $y=Ax_{0}+z$, for any bounded noise setting $\mathcal{B}$,
   the weighted $l_{1}$ minimization (\ref{f2}) can not stably recover the $k-$sparse signal $x_{0}$, i.e., $\widehat{x}\nrightarrow x_{0}$ as $z\rightarrow 0$,
   where $\widehat{x}$ is the solution to (\ref{f2}).
 \end{description}
\end{Theorem}

\begin{Proposition}
\begin{description}
  \item[{\rm (1)}] If $\omega=1$, then $d=1$ and $\gamma=1$. The sufficient condition (\ref{g23})
   of Theorem \ref{t3} is identical to (\ref{g1}) in Theorem \ref{t1}
   and $D_{0}=C_{0}, D_{1}=C_{1}, D'_{0}=C'_{0}, D'_{1}=C'_{1}$.
   Moreover, the condition is sharp if $t\geq \frac{4}{3}$.
  \item[{\rm (2)}]If $\alpha=\frac{1}{2}$, then $d=1, \gamma=1$. The sufficient condition (\ref{g23}) of Theorem \ref{t3} is identical to that of Theorem \ref{t1} with (\ref{g1}) and $D_{0}=C_{0}, D_{1}=C_{1}, D'_{0}=C'_{0}, D'_{1}=C'_{1}$. Moreover, if $t\geq \frac{4}{3}$, the condition is sharp.
  \item[{\rm (3)}] Assume that $0\leq\omega<1$. If $\alpha>\frac{1}{2}$, then $d=1$ and $\gamma<1$. The sufficient condition (\ref{g23}) in Theorem \ref{t3}
  is weaker than (\ref{g1}) in Theorem \ref{t1} and (\ref{g7}) in Remark
  \ref{r1}. Then $D_{0}<C_{0},~ D_{1}<C_{1},~ D_{0}<C''_{0},~ D_{1}<C''_{1}$.
  When $t=2$, the sufficient condition (\ref{g23}) in Theorem \ref{t3}
  is weaker than (\ref{g8}) in Remark \ref{r1}.
\end{description}
\end{Proposition}

\begin{figure*}[htbp!]
\centering
\begin{minipage}[htbp]{0.46\linewidth}
\centering
\includegraphics[width=4in]{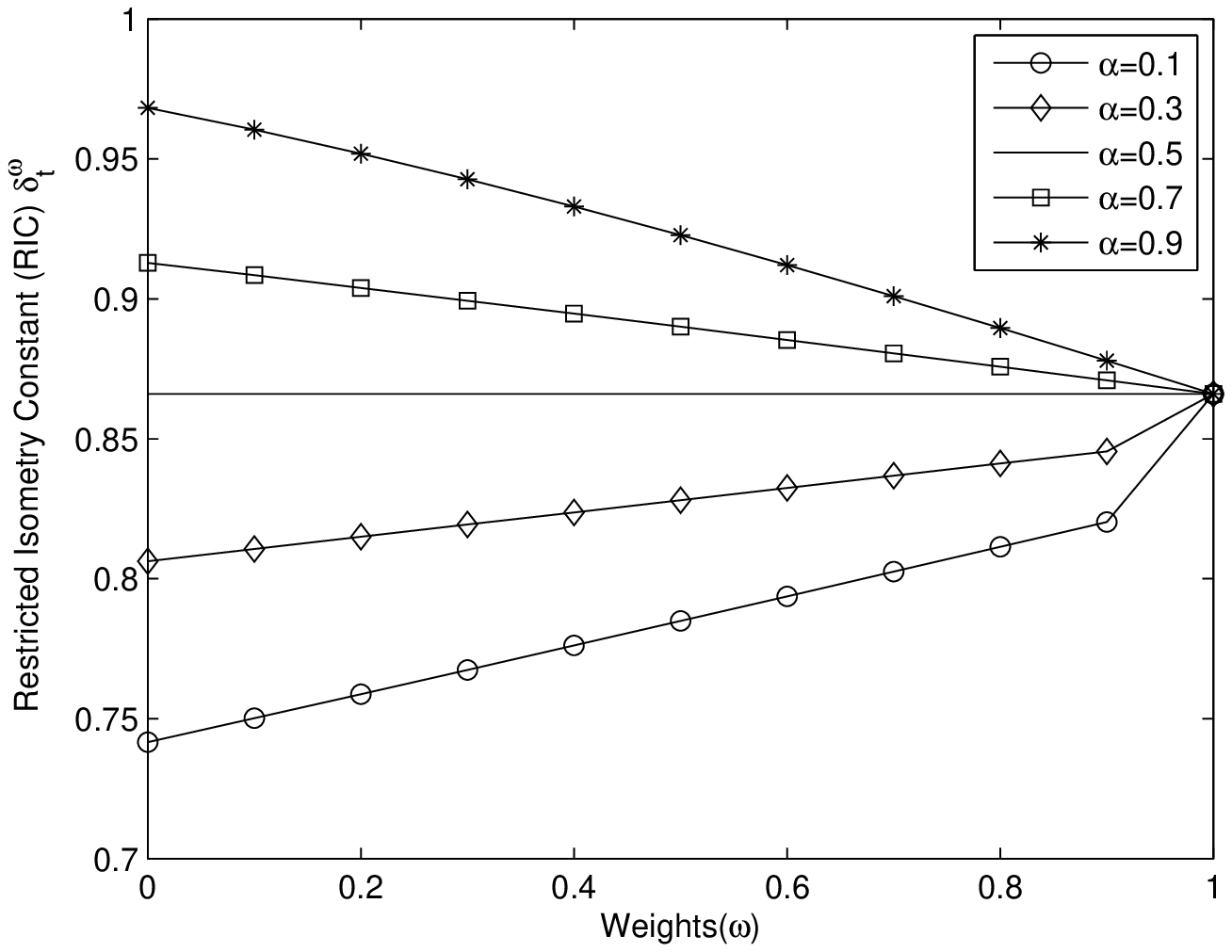}\\[-0.3cm]
{(a) $\delta_{t}^{\omega}$ versus $\omega$}
\end{minipage}
\\
\begin{minipage}[htbp]{0.46\linewidth}
\centering
\includegraphics[width=3.3in]{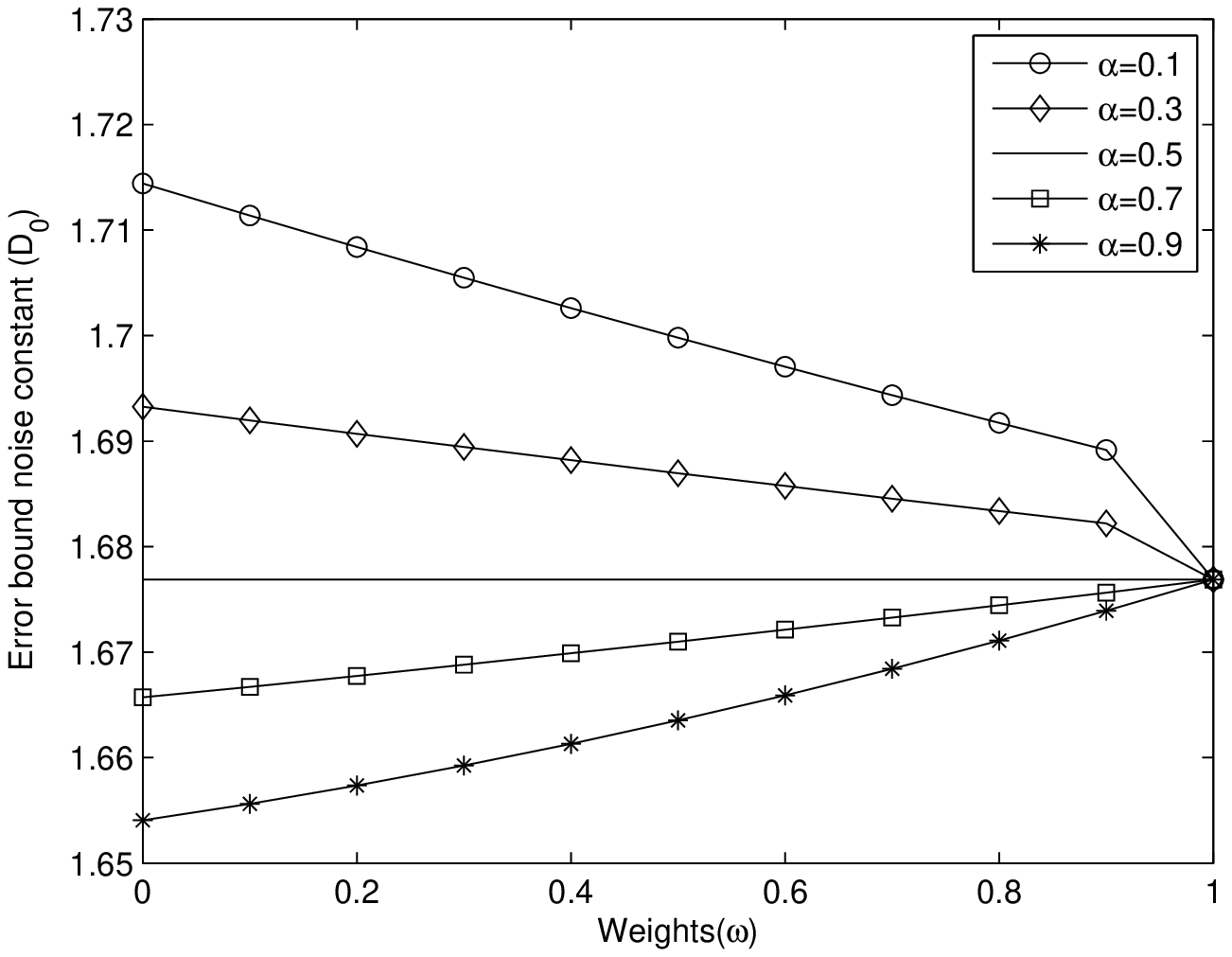}\\[-0.3cm]
{(b) $D_{0}$ versus $\omega$}
\end{minipage}
\begin{minipage}[htbp]{0.46\linewidth}
\centering
\includegraphics[width=3.3in]{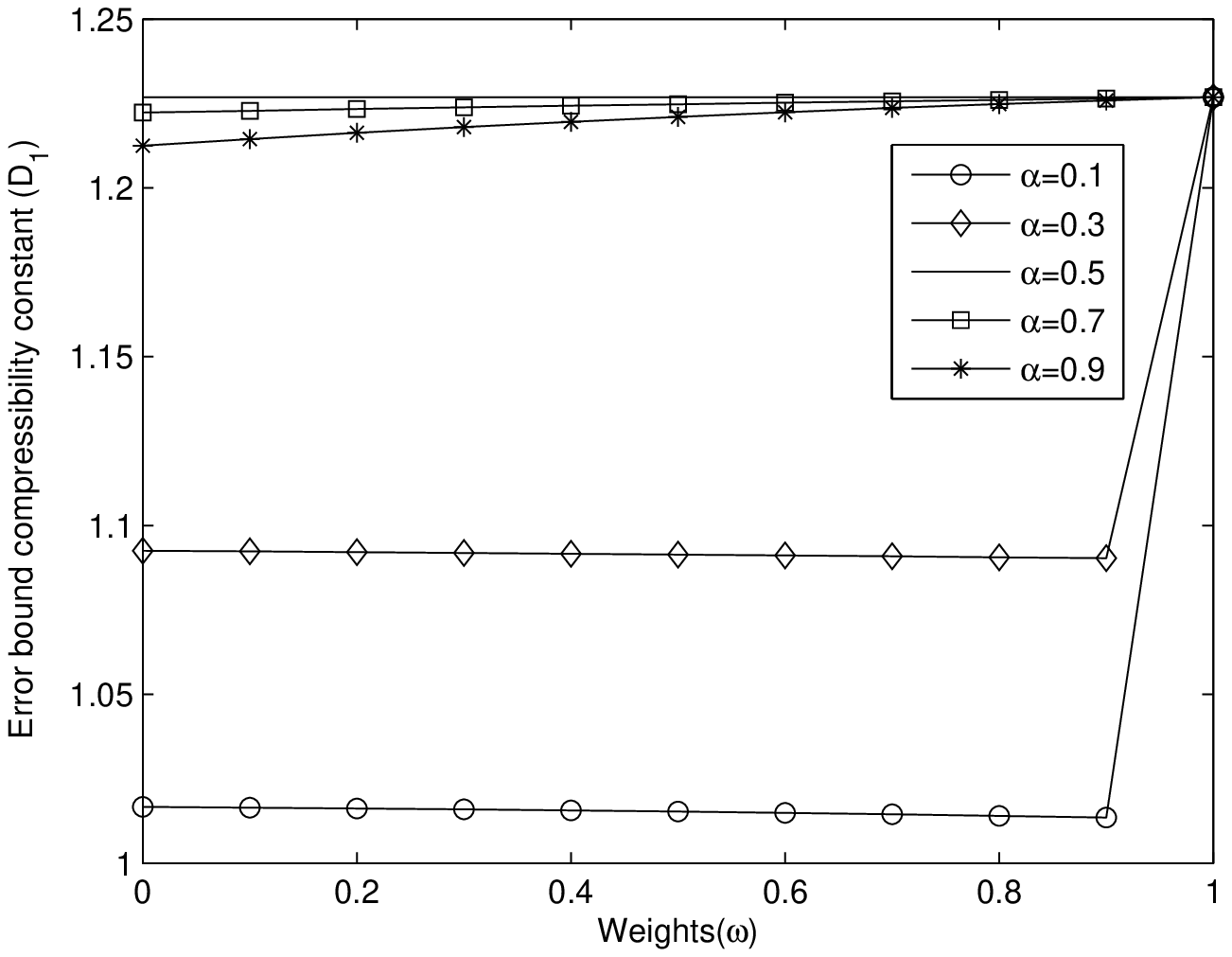}\\[-0.3cm]
{(c) $D_{1}$ versus $\omega$}
\end{minipage}
\\
\begin{minipage}[htbp]{0.46\linewidth}
\centering
\includegraphics[width=3.3in]{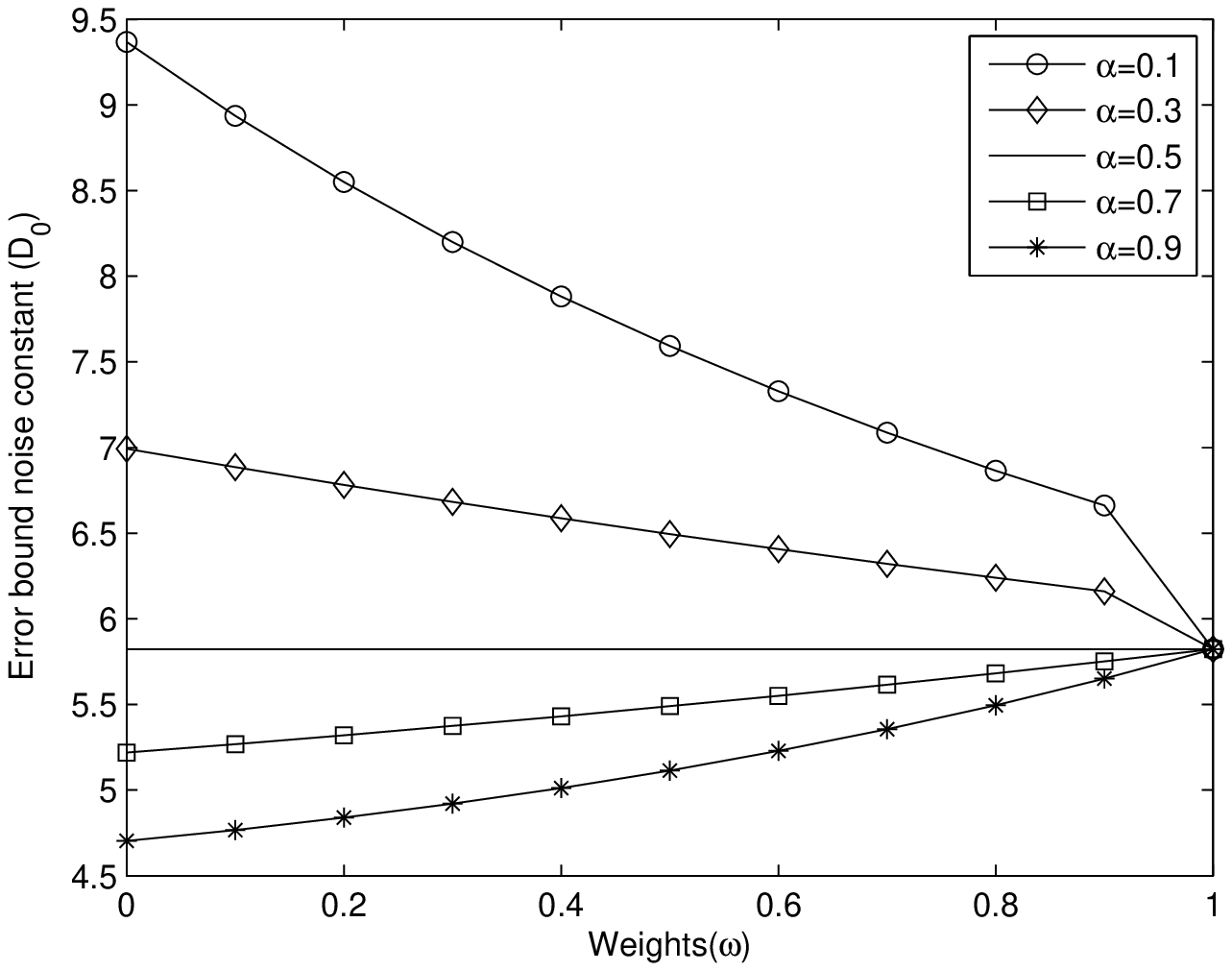}\\[-0.3cm]
{($\mathrm{b}'$) $D_{0}$ versus $\omega$}
\end{minipage}
\begin{minipage}[htbp]{0.46\linewidth}
\centering
\includegraphics[width=3.3in]{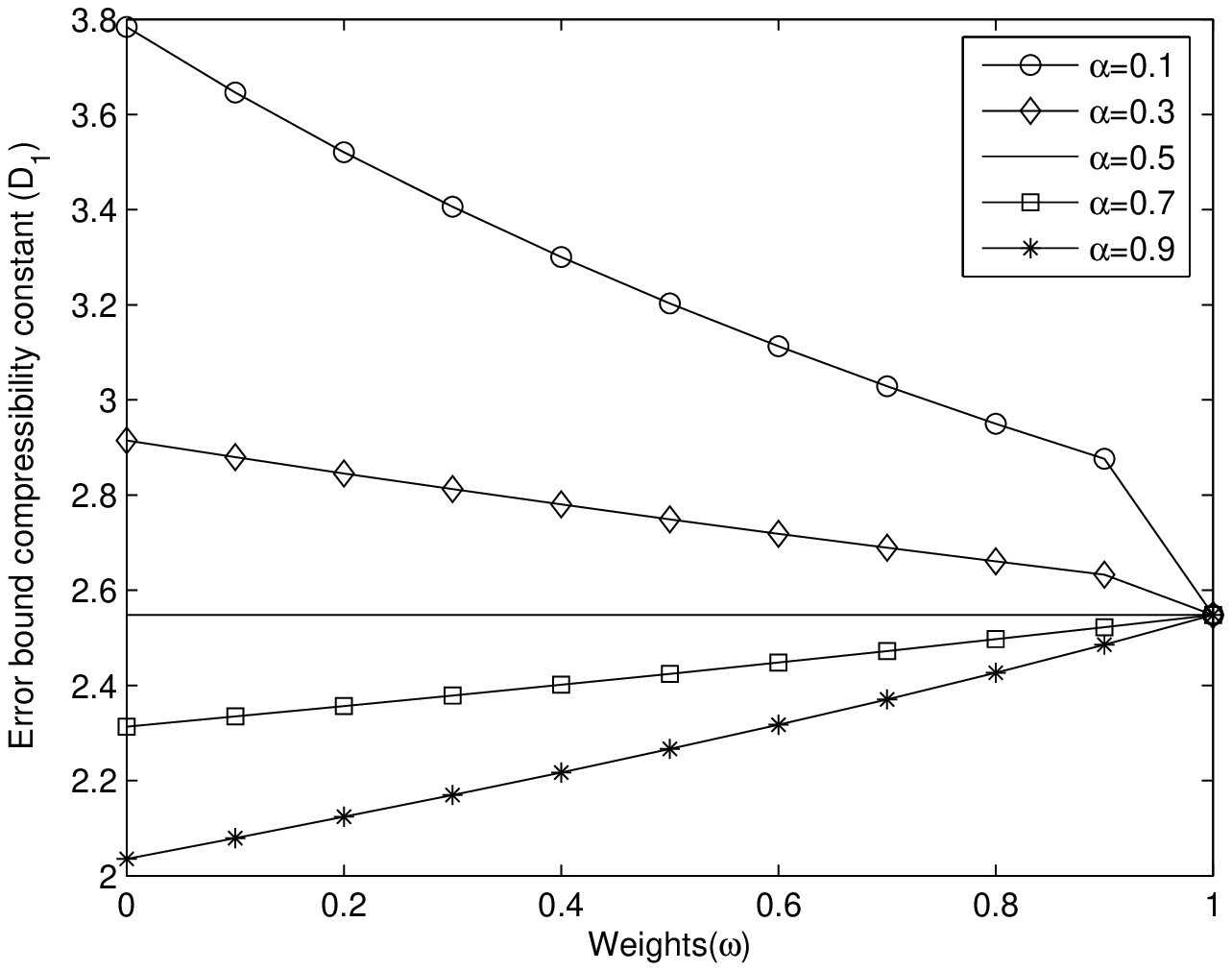}\\[-0.3cm]
{($\mathrm{c}'$) $D_{1}$ versus $\omega$}
\end{minipage}
\\
\caption{\label{fig1}Comparison of the sufficient conditions for
recovery and stability constants for weighted $l_{1}$
 reconstruction with various $\alpha$. In all the figures, we set $t=4$ and $\rho=1$. In (b) and (c),
 we fix $\delta_{tk}=0.1$.
 In ($\mathrm{b}'$) and ($\mathrm{c}'$), we fix $\delta_{tk}=0.6$.  }
\end{figure*}

\begin{figure*}[htbp!]
\centering
\begin{minipage}[htbp]{0.46\linewidth}
\centering
\includegraphics[width=4in]{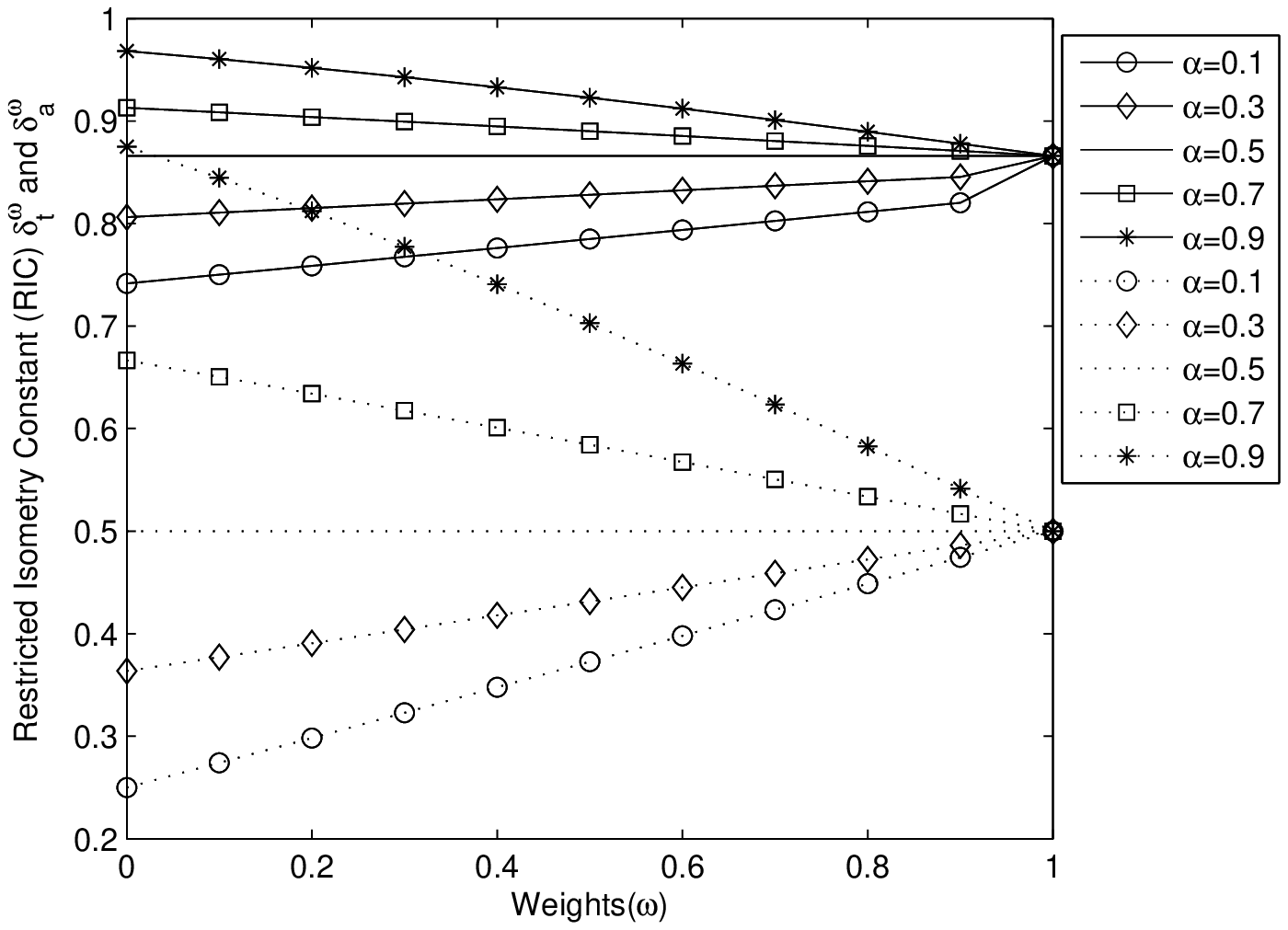}\\[-0.3cm]
{(d) $\delta_{t}^{\omega}$ and $\delta_{a}^{\omega}$ versus
$\omega$}
\end{minipage}
\\
\begin{minipage}[htbp]{0.46\linewidth}
\centering
\includegraphics[width=3.3in]{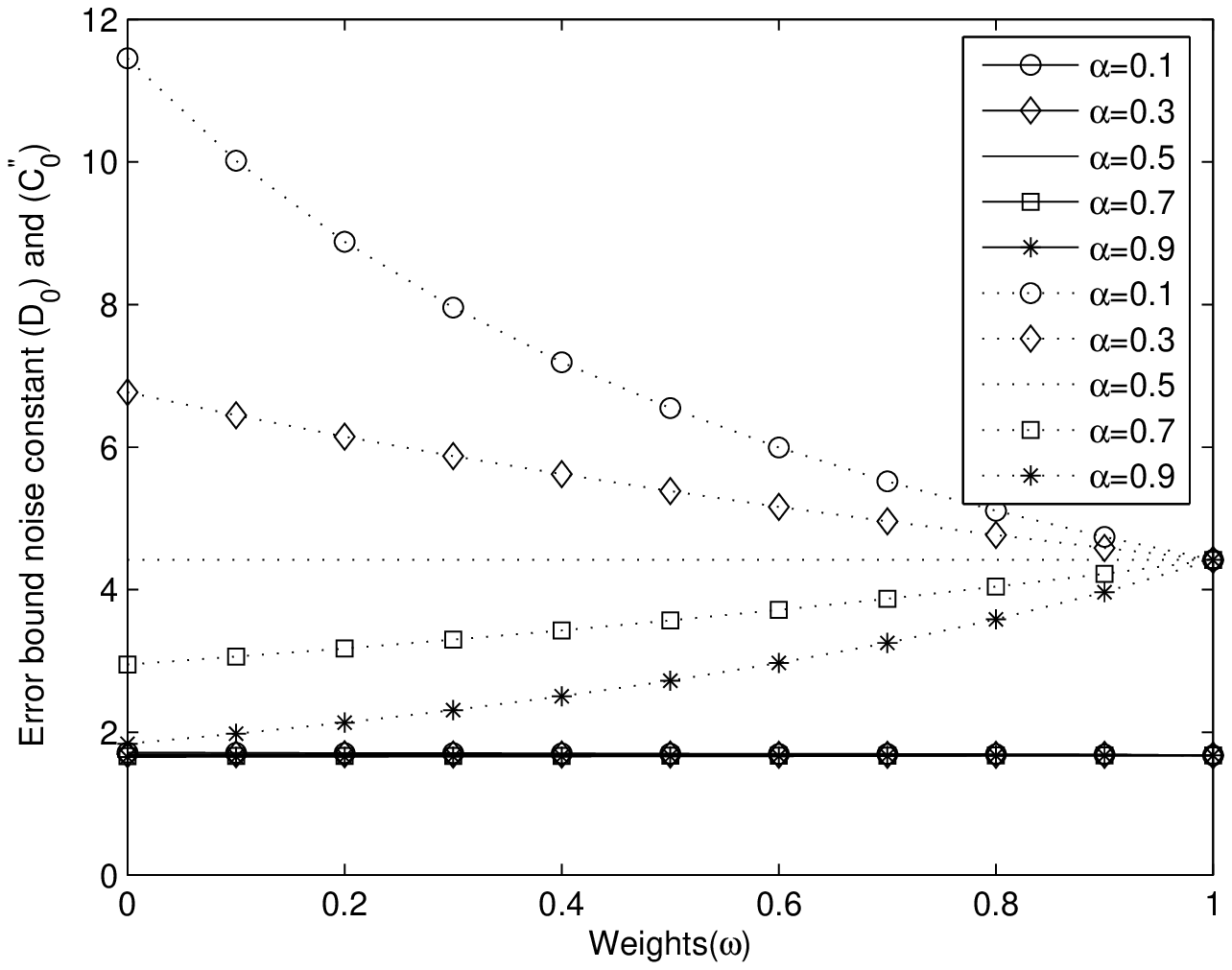}\\[-0.3cm]
{(e) $D_{0}$ and $C''_{0}$ versus $\omega$}
\end{minipage}
\begin{minipage}[htbp]{0.46\linewidth}
\centering
\includegraphics[width=3.3in]{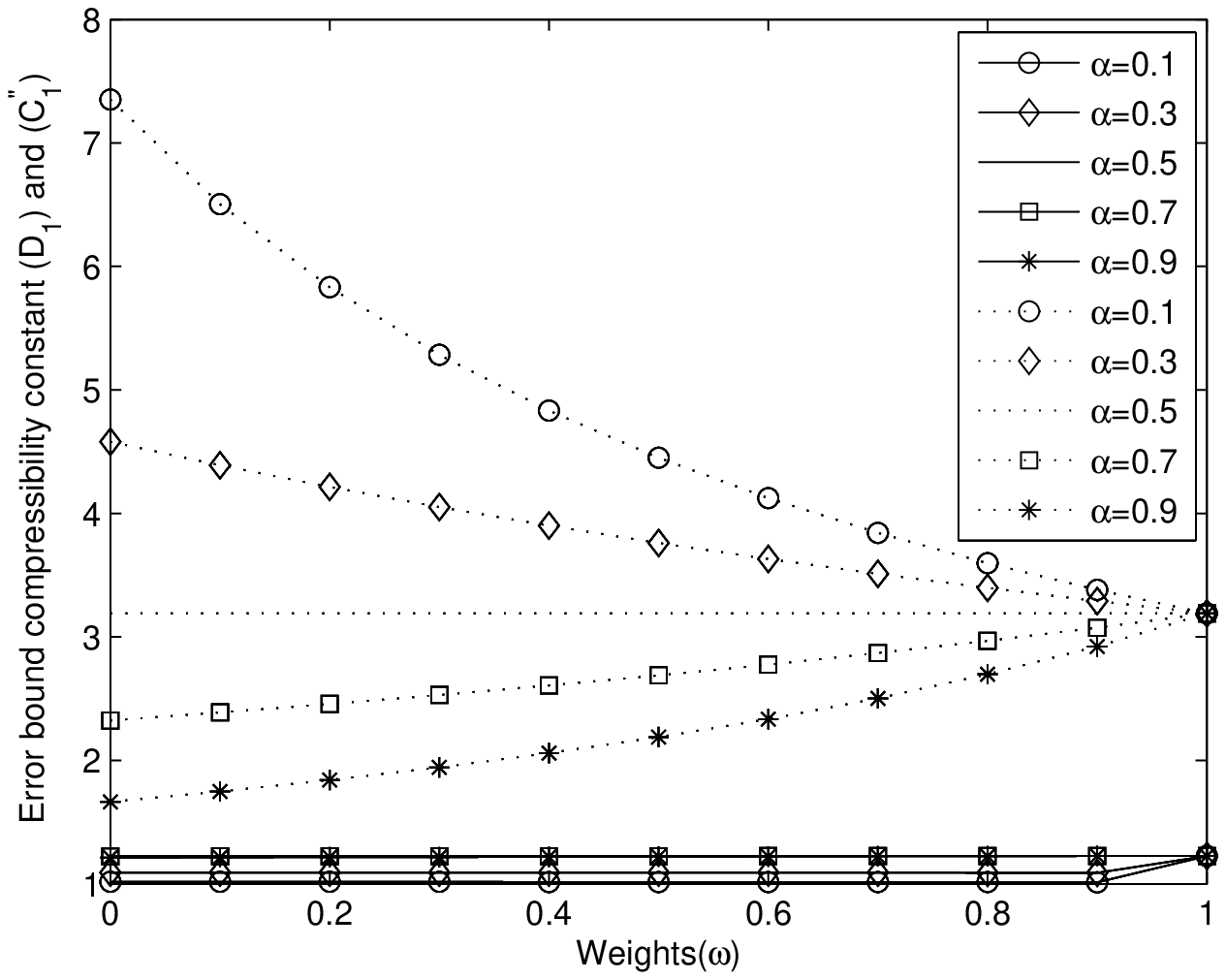}\\[-0.3cm]
{(f) $D_{1}$ and   $C''_{1}$ versus $\omega$}
\end{minipage}
\\
\caption{\label{fig2}Comparison between the bounds of sufficient
recovery conditions $\delta_{t}^{\omega}$ in (\ref{g23}) and
$\delta_{a}^{\omega}$ in (\ref{g7}) as well as stability constants
in (\ref{g24}) and (\ref{g6}) with various $\alpha$. In all the
figures, we set $t=4, ~a=3$ and $\rho=1$. The solid describe our
main results and the dotted describe the results of Friedlander et
al. \cite{FMSY}. In (e) and (f), we fix
$\delta_{tk}=\delta_{(a+1)k}=0.1$ and $\delta_{ak}=0.05$. }
\end{figure*}

Fig. \ref{fig1} illustrates how the sufficient conditions on the RIP
constants given in (\ref{g23}) and the stability constants given in
(\ref{g24}) change with $\omega$ for different values of $\alpha$ in
the case of weighted $l_{1}$ when $t=4$. Note that (\ref{g23})
reduces to (\ref{g1}), and (\ref{g24}) reduces to (\ref{g21}) if
$\omega=1$ or $\alpha=0.5$. In Fig. \ref{fig1} (a), we plot
$\delta_{t}^{\omega}$ versus $\omega$ with different values of
$\alpha$ when $t=4$. We observe that the bound on RIP constant gets
larger as $\alpha$ increases. That is to say, the sufficient
condition on the RIP constant becomes weaker as $\alpha$ increases.
For example, if $90\%$ of the support estimate is accurate and
$\omega=0.4$, we have $\delta_{t}^{\omega}=0.9330$, however
$\delta_{t}^{1}= 0.8660$ of standard $l_{1}$. Figs. \ref{fig1}(b)
and \ref{fig1}($\mathrm{b}'$) show that the constant $D_{0}$
decreases as $\alpha$ increases with $\delta_{tk}=0.1$ and
$\delta_{tk}=0.6$, respectively. But Fig. \ref{fig1}(c) demonstrates
the constant $D_{1}$ with $\alpha\neq0.5$ is smaller than that with
$\alpha=0.5$ when $\delta_{tk}=0.1$. Fig. \ref{fig1}($\mathrm{c}'$)
illustrates that the constant $D_{1}$ decreases as $\alpha$
increases with $\delta_{tk}=0.6$. From above recovery results by
standard $l_{1}$ and weighted $l_{1}$, we see that if the partial
support estimate is more than $50\%$ accurate, i.e. $\alpha>0.5$,
the measurement matrix $A$ for signal recovery by weighted $l_{1}$
satisfies weaker conditions than the analogous conditions for
recovery by standard $l_{1}$. Moreover, we have better upper bounds
when $\alpha
>0.5$ than those of standard $l_{1}$.

Fig. \ref{fig2} compares the sufficient recovery conditions
$\delta_{t}^{\omega}$ in (\ref{g23}) and  $\delta_{a}^{\omega}$ in
(\ref{g7}) as well as stability constants in (\ref{g24}) and
(\ref{g6}) with various $\alpha$ when $t=4, ~a=3$ and
$\delta_{tk}=\delta_{(a+1)k}=0.1$. Here we plot
$\delta_{t}^{\omega}$ and $\delta_{a}^{\omega}$ as well as
(\ref{g24}) and (\ref{g6}) versus $\omega$ with various $\alpha$.
Fig.\ref{fig2}(d) illustrates $\delta_{t}^{\omega}$ is larger than
$\delta_{a}^{\omega}$ under the same support estimate. Moreover,
Figs. \ref{fig2}(e) and \ref{fig2}(f) describe that constants
$D_{0}$ and $D_{1}$ are always smaller than $C''_{0}$ and $C''_{1}$,
respectively. These results state that the sufficient condition
(\ref{g23}) is weaker than (\ref{g7}), and error bound constants
(\ref{g24}) in Theorem \ref{t3}  are better than those (\ref{g6}) in
Theorem \ref{t2}.

\section{Proofs}\label{4}
\begin{proof}[Proof of Theorem \ref{t3}]
Firstly, we show the estimate (\ref{g9}). Let
$\widehat{x}^{l_{2}}=x+h$, where $x$ is the original signal and
$\widehat{x}^{l_{2}}$ is the minimizer of (\ref{f2}) with
(\ref{b1}). Now assume that $tk$ is an integer. We use the following
inequality which has been shown by Friedlander $et~ al.$ (see (21)
in \cite{FMSY}).
\begin{align}\label{g11}
  \|h_{T_{0}^{c}}\|_{1}\leq&\omega\|h_{T_{0}}\|_{1}+(1-\omega)\|h_{T_{0}\cup\widetilde{T}
  \setminus\widetilde{T}_{\alpha}}\|_{1}
  +2\left(\omega\|x_{T_{0}^{c}}\|_{1}+(1-\omega)\|x_{\widetilde{T}^{c}\cap T_{0}^{c}}\|_{1}\right).
\end{align}
Let $\widetilde{T}_{0}=T_{0}\backslash\widetilde{T}_{\alpha}$, and
$T_{1}$ indexes the $ak$ largest in magnitude coefficients of
$h_{\widetilde{T}_{0}^{c}}$, where $|T_{1}|=ak$ and $a=\max\{\alpha,
\beta\}\rho$. Denote
 $$T_{01}=\left\{
               \begin{array}{cc}
                 T_{0}, & \omega=1, \\
                 \widetilde{T}_{0}\cup T_{1}, & 0\leq \omega<1.
               \end{array}
               \right.
               $$
Clearly, $|T_{01}|=dk$ where $$d=\left\{
 \begin{array}{cc}
   1,& \omega=1 \\
   1-\alpha\rho+a, & 0\leq\omega<1
 \end{array}
 \right..$$
From (\ref{g11}) and $d\geq 1$, it is clear that
\begin{align}\label{g12}
   \|h_{-\max(dk)}\|_{1}&\leq\omega\|h_{T_{0}}\|_{1}+(1-\omega)\|h_{T_{0}
   \cup\widetilde{T}\setminus\widetilde{T}_{\alpha}}\|_{1}
  +2\left(\omega\|x_{T_{0}^{c}}\|_{1}+(1-\omega)\|x_{\widetilde{T}^{c}\cap T_{0}^{c}}\|_{1}\right).
\end{align}
Let
\begin{align*}r=&\frac{1}{k}\Big[\omega\|h_{T_{0}}\|_{1}+(1-\omega)\|h_{T_{0}
\cup\widetilde{T}\setminus\widetilde{T}_{\alpha}}\|_{1}
+2\left(\omega\|x_{T_{0}^{c}}\|_{1}+(1-\omega)\|x_{\widetilde{T}^{c}\cap
T_{0}^{c}}\|_{1}\right)\Big].\end{align*}
 We partition $h_{-\max(dk)}$ into two parts, i.e.,
 $h_{-\max(dk)}=h^{(1)}+h^{(2)}$, where
  $h^{(1)}(i)$ equals to $h_{-\max(dk)}(i)$ if $|h_{-\max(dk)}(i)|>\frac{r}{t-d}$ and $0$ else,
  $h^{(2)}(i)$ equals to $h_{-\max(dk)}(i)$ if $|h_{-\max(dk)}(i)|\leq\frac{r}{t-d}$ and $0$ else.

 In view of the above definitions and  (\ref{g12}),
 $$\|h^{(1)}\|_{1}\leq \|h_{-\max{(dk)}}\|_{1}\leq kr.$$
Let
$$\|h^{(1)}\|_{0}=m.$$
From the definition of $h^{(1)}$, it is clear that
\begin{align*}kr\geq \|h^{(1)}\|_{1}
&=\sum\limits_{i\in \mathrm{supp}(h^{(1)})}|h^{(1)}(i)|  \geq
\sum\limits_{i\in \mathrm{supp}(h^{(1)})}\frac{r}{t-d}
=\frac{mr}{t-d}.\end{align*} Namely $m\leq k(t-d).$ Moreover,
$\|h_{\max{(dk)}}+h^{(1)}\|_{0}=dk+m\leq dk+k(t-d)=tk,$ and
\begin{align}\label{g13}
\|h^{(2)}\|_{1}&=\|h_{-\max(dk)}\|_{1}-\|h^{(1)}\|_{1} \leq
kr-\frac{mr}{t-d}  \notag \\
&=(k(t-d)-m)\cdot\frac{r}{t-d},    \notag\\
\|h^{(2)}\|_{\infty}&\leq \frac{r}{t-d}.
\end{align}
By the definition of $\delta_{k}$ and the fact that
\begin{align}\label{g14}
  \|Ah\|_{2}&\leq \|A\widehat{x}^{l_{2}}-Ax\|_{2}
  \leq\|y-A\widehat{x}^{l_{2}}\|_{2}+\|Ax-y\|_{2}\leq 2\varepsilon,
\end{align}
 we obtain
\begin{align}\label{g15}
 \langle A(h_{\max{(dk)}}+h^{(1)}), Ah\rangle
  \leq &\|A(h_{\max{(dk)}}+h^{(1)})\|_{2}\|Ah\|_{2} \notag\\
  \leq &\sqrt{1+\delta_{tk}}\|h_{\max{(dk)}}+h^{(1)}\|_{2}\cdot(2\varepsilon).
\end{align}
Thus, using Lemma \ref{l1} and (\ref{g13}), we have
$h^{(2)}=\sum\limits_{i=1}^{N}\lambda_{i}u_{i},$
$\mathrm{supp}(u_{i})\subseteq \mathrm{supp}(h^{(2)}),
  \|u_{i}\|_{1}=\|h^{(2)}\|_{1}$ and $\|u_{i}\|_{\infty}\leq \frac{r}{t-d},$
where $u_{i}$ is $(k(t-d)-m)-$sparse, namely,
$|\mathrm{supp}(u_{i})|=\|u_{i}\|_{0}\leq k(t-d)-m$. Thus,
\begin{align*}
 \|u_{i}\|_{2}&\leq \sqrt{\|u_{i}\|_{0}}\|u_{i}\|_{\infty}\leq \sqrt{k(t-d)-m}\|u_{i}\|_{\infty} \\
 &\leq \sqrt{k(t-d)}\cdot\frac{r}{t-d}\leq \sqrt{\frac{k}{t-d}}r.
\end{align*}
Take $\beta_{i}=h_{\max{(dk)}}+h^{(1)}+\mu u_{i},$ where $0\leq \mu
\leq 1$. We observe that
\begin{align}\label{g16}
  \sum\limits_{j=1}^{N}&\lambda_{j}\beta_{j}-\frac{1}{2}\beta_{i}
  =h_{\max{(dk)}}+h^{(1)}+\mu h^{(2)}-\frac{1}{2}\beta_{i} \notag\\
  =&(\frac{1}{2}-\mu)(h_{\max{(dk)}}+h^{(1)})-\frac{1}{2}\mu u_{i}+\mu h.
\end{align}
Because $h_{\max{(dk)}}$ is $dk-$sparse, $h^{(1)}$ is $m-$sparse,
and $u_{i}$ is $k(t-d)-m-$sparse, $\beta_{i}$ and
$\sum\limits_{j=1}^{N}\lambda_{j}\beta_{j}-\frac{1}{2}\beta_{i}-\mu
h$ are $tk-$sparse. Let
\begin{align*}
X&=\|h_{\max{(dk)}}+h^{(1)}\|_{2},\\
\gamma&=\omega+(1-\omega)\sqrt{1+\rho-2\alpha\rho}, \\
P&=\frac{2\left(\omega\|x_{T_{0}^{c}}\|_{1}+(1-\omega)\|x_{\widetilde{T}^{c}\cap
T_{0}^{c}}\|_{1}\right)}{\sqrt{k}\gamma}.\end{align*}
 Due to
$|T_{0}\cup\widetilde{T}\backslash\widetilde{T}_{\alpha}|=(1+\rho-2\alpha\rho)k$,
\begin{align}\label{g17}
  \|u_{i}\|_{2}\leq& \sqrt{\frac{k}{t-d}}r \notag\\
=&\sqrt{\frac{k}{t-d}}\cdot\frac{1}{k}\Big[\omega\|h_{T_{0}}\|_{1}+(1-\omega)\|h_{T_{0}
\cup\widetilde{T}\setminus\widetilde{T}_{\alpha}}\|_{1}
+2\left(\omega\|x_{T_{0}^{c}}\|_{1}+(1-\omega)\|x_{\widetilde{T}^{c}\cap
T_{0}^{c}}\|_{1}\right)\Big] \notag\\
=&\frac{\omega\|h_{T_{0}}\|_{1}+(1-\omega)\|h_{T_{0}\cup\widetilde{T}
\setminus\widetilde{T}_{\alpha}}\|_{1}}{\sqrt{k(t-d)}}
+\frac{2\left(\omega\|x_{T_{0}^{c}}\|_{1}+(1-\omega)\|x_{\widetilde{T}^{c}
\cap T_{0}^{c}}\|_{1}\right)}{\sqrt{k(t-d)}} \notag\\
\leq&\frac{\omega\sqrt{k}\|h_{T_{0}}\|_{2}+(1-\omega)\sqrt{(1+\rho-2\alpha\rho)k}
\|h_{T_{0}\cup\widetilde{T}\setminus\widetilde{T}_{\alpha}}\|_{2}}{\sqrt{k(t-d)}}
  +\frac{2\left(\omega\|x_{T_{0}^{c}}\|_{1}+(1-\omega)\|x_{\widetilde{T}^{c}\cap T_{0}^{c}}\|_{1}\right)}{\sqrt{k(t-d)}}\notag\\
\leq&\frac{\omega\|h_{\max{(dk)}}\|_{2}+(1-\omega)\sqrt{1+\rho-2\alpha\rho}\|h_{\max{(dk)}}\|_{2}}{\sqrt{t-d}}
  +\frac{2\left(\omega\|x_{T_{0}^{c}}\|_{1}+(1-\omega)\|x_{\widetilde{T}^{c}\cap T_{0}^{c}}\|_{1}\right)}{\sqrt{k(t-d)}}  \notag\\
=&\frac{\left(\omega+(1-\omega)\sqrt{1+\rho-2\alpha\rho}\right)\|h_{\max{(dk)}}\|_{2}}{\sqrt{t-d}}
  +\frac{2\left(\omega\|x_{T_{0}^{c}}\|_{1}+(1-\omega)\|x_{\widetilde{T}^{c}\cap T_{0}^{c}}\|_{1}\right)}{\sqrt{k(t-d)}} \notag\\
\leq& \frac{\gamma}{\sqrt{t-d}}\|h_{\max{(dk)}}+h^{(1)}\|_{2}
  +\frac{2\left(\omega\|x_{T_{0}^{c}}\|_{1}+(1-\omega)\|x_{\widetilde{T}^{c}\cap T_{0}^{c}}\|_{1}\right)}{\sqrt{k(t-d)}} \notag\\
  =&\frac{\gamma}{\sqrt{t-d}}(X+P).
\end{align}

We use the following identity (see (25) in \cite{CZ})
\begin{align}\label{g18}
 \sum\limits_{i=1}^{N}\lambda_{i}\Big\|A\Big(\sum\limits_{j=1}^{N}\lambda_{j}\beta_{j}-\frac{1}{2}\beta_{i}\Big)\Big\|_{2}^{2}
 =\sum\limits_{i=1}^{N}\frac{\lambda_{i}}{4}\|A\beta_{i}\|_{2}^{2}.
\end{align}

Combining (\ref{g15}) and (\ref{g16}), we can estimate the left hand
side of (\ref{g18})
\begin{align*}
\sum\limits_{i=1}^{N}&\lambda_{i}\Big\|A\Big(\sum\limits_{j=1}^{N}\lambda_{j}\beta_{j}-\frac{1}{2}\beta_{i}\Big)\Big\|_{2}^{2}\\
 &=\sum\limits_{i=1}^{N}\lambda_{i}\Big\|A\Big[(\frac{1}{2}-\mu)(h_{\max{(dk)}}+h^{(1)})-\frac{1}{2}\mu u_{i}+\mu h\Big]\Big\|_{2}^{2} \\
&=\sum\limits_{i=1}^{N}\lambda_{i}\Big \|A
 \Big[(\frac{1}{2}-\mu)(h_{\max{(dk)}}+h^{(1)})-\frac{1}{2}\mu u_{i}\Big]\Big\|_{2}^{2}  \\
&+2\left\langle A\left(
(\frac{1}{2}-\mu)(h_{\max{(dk)}}+h^{(1)})-\frac{1}{2}\mu
h^{(2)}\right), \mu Ah\right\rangle
+\mu^{2}\|Ah\|_{2}^{2} \\
&=\sum\limits_{i=1}^{N}\lambda_{i}\Big \|A
 \Big[(\frac{1}{2}-\mu)(h_{\max{(dk)}}+h^{(1)})-\frac{1}{2}\mu u_{i}\Big]\Big\|_{2}^{2}  \\
 &+\mu(1-\mu)\langle A(h_{\max{(dk)}}+h^{(1)}), Ah\rangle  \\
&\leq(1+\delta_{tk})\sum\limits_{i=1}^{N}\lambda_{i}\Big\|(\frac{1}{2}-\mu)(h_{\max{(dk)}}+h^{(1)})-\frac{1}{2}\mu u_{i}\Big\|_{2}^{2}  \\
&+\mu(1-\mu)\sqrt{1+\delta_{tk}}\|h_{\max{(dk)}}+h^{(1)}\|_{2}\cdot(2\varepsilon) \\
&=(1+\delta_{tk})\sum\limits_{i=1}^{N}\lambda_{i}\Big[(\frac{1}{2}-\mu)^{2}\|h_{\max{(dk)}}+h^{(1)}\|_{2}^{2}
   +\frac{\mu^{2}}{4}
\|u_{i}\|_{2}^{2}\Big]\\
&+\mu(1-\mu)\sqrt{1+\delta_{tk}}\|h_{\max{(dk)}}+h^{(1)}\|_{2}\cdot(2\varepsilon).
\end{align*}

On the other hand, in view of the expression of $\beta_{i}$,
\begin{align*}
\sum\limits_{i=1}^{N}\frac{\lambda_{i}}{4}\|A\beta_{i}\|_{2}^{2}
&=\sum\limits_{i=1}^{N}\frac{\lambda_{i}}{4}\|A(h_{\max{(dk)}}+h^{(1)}+\mu u_{i})\|_{2}^{2} \\
&\geq\sum\limits_{i=1}^{N}\frac{\lambda_{i}}{4}(1-\delta_{tk})\|h_{\max{(dk)}}+h^{(1)}+\mu u_{i}\|_{2}^{2}\\
&= (1-\delta_{tk})\sum\limits_{i=1}^{N}\frac{\lambda_{i}}{4}
\Big(\|h_{\max{(dk)}}+h^{(1)}\|_{2}^{2}+\mu^{2}\|u_{i}\|_{2}^{2}\Big).
\end{align*}
It follows from the above two inequalities and (\ref{g17}) that
\begin{align*}
  0=&\sum\limits_{i=1}^{N}\lambda_{i}\Big\|A\Big(\sum\limits_{j=1}^{N}\lambda_{j}\beta_{j}-\frac{1}{2}\beta_{i}\Big)\Big\|_{2}^{2}
-\sum\limits_{i=1}^{N}\frac{\lambda_{i}}{4}\|A\beta_{i}\|_{2}^{2}  \\
\leq&(1+\delta_{tk})\sum\limits_{i=1}^{N}\lambda_{i}\Big[(\frac{1}{2}-\mu)^{2}\|h_{\max{(dk)}}+h^{(1)}\|_{2}^{2}
   +\frac{\mu^{2}}{4} \|u_{i}\|_{2}^{2}\Big] \\
&+\mu(1-\mu)\sqrt{1+\delta_{tk}}\|h_{\max{(dk)}}+h^{(1)}\|_{2}\cdot(2\varepsilon) \\
&-(1-\delta_{tk})\sum\limits_{i=1}^{N}\frac{\lambda_{i}}{4}\Big(\|h_{\max{(dk)}}+h^{(1)}\|_{2}^{2}+\mu^{2}\|u_{i}\|_{2}^{2}\Big)  \\
=&\sum\limits_{i=1}^{N}\lambda_{i}\bigg\{\left((1+\delta_{tk})(\frac{1}{2}-\mu)^{2}-\frac{1}{4}(1-\delta_{tk})\right)\\
&\cdot\|h_{\max{(dk)}}+h^{(1)}\|_{2}^{2}+\frac{1}{2}\delta_{tk}\mu^{2}\|u_{i}\|_{2}^{2}\bigg\}\\
&+\mu(1-\mu)\sqrt{1+\delta_{tk}}\|h_{\max{(dk)}}+h^{(1)}\|_{2}\cdot(2\varepsilon)\\
\leq&\left[(1+\delta_{tk})(\frac{1}{2}-\mu)^{2}-\frac{1}{4}(1-\delta_{tk})+\frac{\delta_{tk}\mu^{2}\gamma^{2}}{2(t-d)}\right]
X^{2}  \\
&+\left[\mu(1-\mu)\sqrt{1+\delta_{tk}}\cdot(2\varepsilon)+\frac{\delta_{tk}\mu^{2}\gamma^{2}P}{t-d}\right]X
+\frac{\delta_{tk}\mu^{2}\gamma^{2}P^{2}}{2(t-d)}\\
=&\Big[(\mu^{2}-\mu)+\Big(\frac{1}{2}-\mu+(1+\frac{\gamma^{2}}{2(t-d)})\mu^{2}\Big)\delta_{tk}\Big]X^{2}\\
&+\Big[\mu(1-\mu)\sqrt{1+\delta_{tk}}\cdot(2\varepsilon)+\frac{\delta_{tk}\mu^{2}\gamma^{2}P}{t-d}\Big]X
  +\frac{\delta_{tk}\mu^{2}\gamma^{2}P^{2}}{2(t-d)}.
\end{align*}

Taking $\mu=\frac{\sqrt{(t-d)(t-d+\gamma^{2})}-(t-d)}{\gamma^{2}}$,
we obtain
\begin{align*}
  -&\frac{t-d+\gamma^{2}}{t-d}\mu^{2}\left(\sqrt{\frac{t-d}{t-d+\gamma^{2}}}-\delta_{tk}\right)X^{2}  \\
+&\bigg(\frac{t-d+\gamma^{2}}{t-d}\mu^{2}\sqrt{\frac{t-d}{t-d+\gamma^{2}}}\sqrt{1+\delta_{tk}}\cdot(2\varepsilon)
+\frac{\delta_{tk}\mu^{2}\gamma^{2}P}{t-d}\bigg)X
+\frac{\delta_{tk}\mu^{2}\gamma^{2}P^{2}}{2(t-d)}\geq0.
\end{align*}
Namely,
\begin{align*}
 \frac{\mu^{2}}{t-d}&\Big[-(t-d+\gamma^{2})\Big(\sqrt{\frac{t-d}{t-d+\gamma^{2}}}-\delta_{tk}\Big)X^{2}  \\
&+\Big(\sqrt{(t-d)(t-d+\gamma^{2})(1+\delta_{tk})}\cdot(2\varepsilon)\\
&+\delta_{tk}\gamma^{2}P\Big)X
+\frac{\delta_{tk}\gamma^{2}P^{2}}{2}\Big]\geq 0,
\end{align*}
which is a second-order inequality for $X$. Hence, we have
\begin{align*}
 &X\leq\bigg\{\Big( 2\varepsilon\sqrt{(t-d)(t-d+\gamma^{2})(1+\delta_{tk})}+\delta_{tk}\gamma^{2}P \Big)\\
 &+\Big[\Big( 2\varepsilon\sqrt{(t-d)(t-d+\gamma^{2})(1+\delta_{tk})}+\delta_{tk}\gamma^{2}P \Big)^{2}\\
 &+2(t-d+\gamma^{2})\Big( \sqrt{\frac{t-d}{t-d+\gamma^{2}}}-\delta_{tk} \Big)\delta_{tk}\gamma^{2}P^{2} \Big]^{1/2}\bigg\} \\
&\cdot \Big(2(t-d+\gamma^{2})\sqrt{\frac{t-d}{t-d+\gamma^{2}}}-\delta_{tk})\Big)^{-1}\\
&\leq\frac{\sqrt{(t-d)(t-d+\gamma^{2})(1+\delta_{tk})}}{(t-d+\gamma^{2})(\sqrt{\frac{t-d}{t-d+\gamma^{2}}}-\delta_{tk})}(2\varepsilon) \\
&+\frac{2\delta_{tk}\gamma^{2}+\sqrt{2(t-d+\gamma^{2})(\sqrt{\frac{t-d}{t-d+\gamma^{2}}}-\delta_{tk})\delta_{tk}\gamma^{2}}}
{2(t-d+\gamma^{2})(\sqrt{\frac{t-d}{t-d+\gamma^{2}}}-\delta_{tk})}P.
\end{align*}
From (\ref{g11}) and the representation of $P$, it is clear that
\begin{align*}
 \|h_{-\max(dk)}\|_{1}
\leq&\|h_{\max(dk)}\|_{1}
 +2\left(\omega\|x_{T_{0}^{c}}\|_{1}+(1-\omega)\|x_{\widetilde{T}^{c}\cap T_{0}^{c}}\|_{1}\right)\\
=&\|h_{\max(dk)}\|_{1}+P\sqrt{k}\gamma.
\end{align*}
It follows from Lemma \ref{l2} that
$$ \|h_{-\max(dk)}\|_{2}\leq\|h_{\max(dk)}\|_{2}+\frac{P\gamma}{\sqrt{d}}.$$
Thus, we have the estimate of $\|h\|_{2}$ by the above inequalities
\begin{align*}
&\|h\|_{2}=\sqrt{\|h_{\max(dk)}\|_{2}^{2}+\|h_{-\max(dk)}\|_{2}^{2}}  \\
\leq&\sqrt{\|h_{\max(dk)}\|_{2}^{2}+\left(\|h_{\max(dk)}\|_{2}+\frac{P\gamma}{\sqrt{d}}\right)^{2}} \\
\leq&\sqrt{2}\|h_{\max(dk)}\|_{2}+\frac{P\gamma}{\sqrt{d}} \\
\leq&\sqrt{2}\|h_{\max(dk)}+h^{(1)}\|_{2}+\frac{P\gamma}{\sqrt{d}} \\
=&\sqrt{2}X+\frac{P\gamma}{\sqrt{d}} \\
\leq&\frac{\sqrt{2(t-d)(t-d+\gamma^{2})(1+\delta_{tk})}}{(t-d+\gamma^{2})(\sqrt{\frac{t-d}{t-d+\gamma^{2}}}-\delta_{tk})}(2\varepsilon) \\
+&\Bigg(\frac{\sqrt{2}\delta_{tk}\gamma^{2}+\sqrt{(t-d+\gamma^{2})(\sqrt{\frac{t-d}{t-d+\gamma^{2}}}-\delta_{tk})\delta_{tk}\gamma^{2}}}
{(t-d+\gamma^{2})(\sqrt{\frac{t-d}{t-d+\gamma^{2}}}-\delta_{tk})}
+\frac{\gamma}{\sqrt{d}}\Bigg)P \\
=&\frac{\sqrt{2(t-d)(t-d+\gamma^{2})(1+\delta_{tk})}}{(t-d+\gamma^{2})(\sqrt{\frac{t-d}{t-d+\gamma^{2}}}-\delta_{tk})}(2\varepsilon) \\
+&\Bigg(\frac{\sqrt{2}\delta_{tk}\gamma^{2}+\sqrt{(t-d+\gamma^{2})(\sqrt{\frac{t-d}{t-d+\gamma^{2}}}-\delta_{tk})\delta_{tk}\gamma^{2}}}
{(t-d+\gamma^{2})(\sqrt{\frac{t-d}{t-d+\gamma^{2}}}-\delta_{tk})}
+\frac{\gamma}{\sqrt{d}}\Bigg)
\frac{2\left(\omega\|x_{T_{0}^{c}}\|_{1}+(1-\omega)\|x_{\widetilde{T}^{c}\cap T_{0}^{c}}\|_{1}\right)}{\sqrt{k}\gamma} \\
=&\frac{\sqrt{2(t-d)(t-d+\gamma^{2})(1+\delta_{tk})}}{(t-d+\gamma^{2})(\sqrt{\frac{t-d}{t-d+\gamma^{2}}}-\delta_{tk})}(2\varepsilon) \\
+&\Bigg(\frac{\sqrt{2}\delta_{tk}\gamma+\sqrt{(t-d+\gamma^{2})(\sqrt{\frac{t-d}{t-d+\gamma^{2}}}-\delta_{tk})\delta_{tk}}}
{(t-d+\gamma^{2})(\sqrt{\frac{t-d}{t-d+\gamma^{2}}}-\delta_{tk})}
+\frac{1}{\sqrt{d}}\Bigg)
\frac{2\left(\omega\|x_{T_{0}^{c}}\|_{1}+(1-\omega)\|x_{\widetilde{T}^{c}\cap
T_{0}^{c}}\|_{1}\right)}{\sqrt{k}}.
\end{align*}

If $tk$ is not an integer, taking $t'=\lceil tk\rceil/k$, then $t'k$
is an integer and $t< t'$. Thus we have
$$\delta_{t'k}=\delta_{tk}<\sqrt{\frac{t-d}{t-d+\gamma^{2}}}<\sqrt{\frac{t'-d}{t'-d+\gamma^{2}}}.$$
Then we can prove the result the same as the proof above by working
on $\delta_{t'k}$. So, we obtain (\ref{g9}).

Next, we prove (\ref{g10}). The proof of (\ref{g10}) is similar to
the proof of (\ref{g9}). We only need to replace (\ref{g14}) and
(\ref{g15}) with the following (\ref{g19}) and (\ref{g20}),
respectively. We also can get (\ref{g10}).

\begin{align}\label{g19}
  \|A^{T}Ah\|_{\infty}  &= \|A^{T}A(\widehat{x}^{DS}-x)\|_{\infty} \notag\\
    \leq& \|A^{T}(A\widehat{x}^{DS}-y)\|_{\infty}+\|A^{T}(y-Ax)\|_{\infty} \notag\\
   \leq& 2\varepsilon,
\end{align}
\begin{align}\label{g20}
  \langle A(h_{\max{(dk)}}+h^{(1)}), Ah\rangle  &=\langle h_{\max{(dk)}}+h^{(1)}, ~A^{T}Ah\rangle \notag\\
    \leq& \|h_{\max{(dk)}}+h^{(1)}\|_{1}\|A^{T}Ah\|_{\infty} \notag\\
  \leq & \sqrt{tk}\|h_{\max{(dk)}}+h^{(1)}\|_{2}\cdot(2\varepsilon).
\end{align}

This completes the proof of Theorem \ref{t3}.
\end{proof}

\begin{proof}[Proof of Theorem \ref{t4}]
For $d=1$, we have
$\gamma=\omega+(1-\omega)\sqrt{1+\rho-2\alpha\rho}\leq 1$. Moreover,
for all $\varepsilon>0$ and $k\geq \frac{6}{\varepsilon}$, we define
$$m'=\frac{1+\sqrt{1-\gamma^{2}}}{\gamma^{2}}\left(t-1+\sqrt{(t-1)(t-1+\gamma^{2})}\right)k$$
and $N\geq k+m'$. Since $t\geq
1+\frac{\left(1-\sqrt{1-\gamma^{2}}\right)^{2}}{\gamma^{2}+2\left(1-\sqrt{1-\gamma^{2}}\right)}$,
we obtain $m'\geq k$. Let $m$ be the largest integer strictly
smaller than $m'$, then $m<m'$ and $m'-m\leq 1$. We take
\begin{align*}x_{1}=\frac{1}{\sqrt{k+\frac{mk^{2}}{m'^{2}}}}(\underbrace{1,\ldots,1}_{k-\alpha\rho k},
\underbrace{-\frac{k}{m'},\ldots,-\frac{k}{m'}}_{\rho
k},\underbrace{1, \ldots, 1}_{\alpha\rho k},
\underbrace{-\frac{k}{m'},\ldots,-\frac{k}{m'}}_{ m-\rho
k},0,\ldots,0)\in\mathbb{R}^{N},\end{align*} if $m>\rho k$; or take
\begin{align*}x_{1}=\frac{1}{\sqrt{k+\frac{mk^{2}}{m'^{2}}}}(\underbrace{1,\ldots,1}_{k-\alpha\rho k},
\underbrace{\overbrace{-\frac{k}{m'},\ldots,-\frac{k}{m'}}^{m},
0,\ldots,0}_{\rho k},
 \underbrace{1, \ldots, 1}_{\alpha\rho k},
0,\ldots,0)\in\mathbb{R}^{N},\end{align*} if $m\leq\rho k$. It is
easy to know $\|x_{1}\|_{2}=1.$ Define the linear map $A:
\mathbb{R}^{N}\rightarrow\mathbb{R}^{N}$ by
\begin{align*}
 Ax&=\sqrt{1+\sqrt{\frac{t-d}{t-d+\gamma^{2}}}}(x-\langle x_{1}, x\rangle x_{1})\\
 &=\sqrt{1+\sqrt{\frac{t-1}{t-1+\gamma^{2}}}}(x-\langle x_{1}, x\rangle x_{1}),
\end{align*}
for all $x\in\mathbb{R}^{N}.$ Then for any $\lceil tk\rceil-$sparse
vector $x$, we get
$$\|Ax\|_{2}^{2} =\left(1+\sqrt{\frac{t-1}{t-1+\gamma^{2}}}\right)\left(\|x\|_{2}^{2}-|\langle x_{1}, x\rangle|^{2}\right).$$
Hence, using Cauchy-Schwarz inequality and the fact that $m'\geq k$,
$m'-m\leq 1$ and
\begin{align*}
 \frac{m'^{2}+k^{2}(t-1)}{m'^{2}+m'k}=\frac{2\sqrt{t-1}(\sqrt{t-1+\gamma^{2}}-\sqrt{t-1})}{\gamma^{2}},
    \end{align*}
we have
\begin{align*}
 0\leq&|\langle x_{1}, x\rangle|^{2}\leq\|x\|_{2}^{2}\cdot\sum\limits_{i\in {\rm supp}(x)}|x_{1}(i)|^{2} \\
 \leq& \|x\|_{2}^{2}\cdot \|x_{1,\max{(\lceil tk\rceil)}}\|_{2}^{2}  \\
 =&\|x\|_{2}^{2}\cdot\frac{m'^{2}+k(\lceil tk\rceil-k)}{m'^{2}+mk} \\
\leq&\frac{m'^{2}+k^{2}(t-1)+k}{m'^{2}+mk}\|x\|_{2}^{2}\\
=&\frac{m'^{2}+k^{2}(t-1)+k}{m'^{2}+m'k}\cdot\frac{m'^{2}+m'k}{m'^{2}+mk}\|x\|_{2}^{2}\\
=&\frac{m'^{2}+k^{2}(t-1)+k}{m'^{2}+m'k}\cdot\frac{1}{1-\frac{k(m'-m)}{m'^{2}+m'k}}\|x\|_{2}^{2}\\
=&\frac{m'^{2}+k^{2}(t-1)}{m'^{2}+m'k}\cdot\frac{m'^{2}+k^{2}(t-1)+k}{m'^{2}+k^{2}(t-1)}\cdot\frac{1}{1-\frac{k(m'-m)}{m'^{2}+m'k}}\|x\|_{2}^{2}\\
\leq&\frac{2\sqrt{t-1}(\sqrt{t-1+\gamma^{2}}-\sqrt{t-1})}{\gamma^{2}}\cdot(1+\frac{1}{tk})\cdot\frac{1}{1-\frac{1}{2k}}\|x\|_{2}^{2}\\
\leq&\frac{2\sqrt{t-1}(\sqrt{t-1+\gamma^{2}}-\sqrt{t-1})}{\gamma^{2}}\cdot(1+\frac{3}{k})\|x\|_{2}^{2}\\
\leq&\left(\frac{2\sqrt{t-1}(\sqrt{t-1+\gamma^{2}}-\sqrt{t-1})}{\gamma^{2}}+\frac{3}{k}\right)\|x\|_{2}^{2}.
\end{align*}
Consequently,
\begin{align*}
      &\left(1+\sqrt{\frac{t-1}{t-1+\gamma^{2}}}+\varepsilon\right)\|x\|_{2}^{2} \\
      \geq&\left(1+\sqrt{\frac{t-1}{t-1+\gamma^{2}}}\right)\|x\|_{2}^{2}\geq\|Ax\|_{2}^{2}\\
      \geq&\left(1+\sqrt{\frac{t-1}{t-1+\gamma^{2}}}\right)\cdot\left(1-\frac{2\sqrt{t-1}(\sqrt{t-1+\gamma^{2}}-\sqrt{t-1})}{\gamma^{2}}-\frac{3}{k}\right)\|x\|_{2}^{2}\\
      =&\Bigg[\left(1+\sqrt{\frac{t-1}{t-1+\gamma^{2}}}\right)\cdot\left(1-\frac{2\sqrt{t-1}(\sqrt{t-1+\gamma^{2}}-\sqrt{t-1})}{\gamma^{2}}\right)\\
      &-\left(1+\sqrt{\frac{t-1}{t-1+\gamma^{2}}}\right)\frac{3}{k} \Bigg]\|x\|_{2}^{2}\\
      =&\left[1-\sqrt{\frac{t-1}{t-1+\gamma^{2}}}-\left(1+\sqrt{\frac{t-1}{t-1+\gamma^{2}}}\right)\frac{3}{k} \right]\|x\|_{2}^{2}\\
      \geq&\left(1-\sqrt{\frac{t-1}{t-1+\gamma^{2}}}-\varepsilon\right)\|x\|_{2}^{2},
     \end{align*}
which deduces $\delta_{tk}\leq
\sqrt{\frac{t-1}{t-1+\gamma^{2}}}+\varepsilon$. Next, we define
\begin{align*}
  &x_{0}=(\overbrace{1,\ldots,1}^{k-\alpha\rho k}, \overbrace{0,\ldots, 0}^{\rho k}, \overbrace{1,\ldots,1}^{\alpha\rho k},0,\ldots,0)\in\mathbb{R}^{N}, \\
  &\eta_{0}=(\underbrace{0,\ldots, 0}_{k-\alpha\rho k}, \underbrace{\frac{k}{m'},\ldots,\frac{k}{m'}}_{\rho k},
  \underbrace{0,\ldots, 0}_{\alpha\rho k},\underbrace{\frac{k}{m'},\ldots,\frac{k}{m'}}_{m-\rho k}, 0,\ldots,0)\in\mathbb{R}^{N},  {\rm if}~ m>\rho
  k,\\
   {\rm or} ~&\eta_{0}=(\underbrace{0,\ldots, 0}_{k-\alpha\rho k}, \underbrace{\overbrace{\frac{k}{m'},\ldots,\frac{k}{m'}}^{m},0,\ldots,0}_{\rho k},
\underbrace{0,\ldots, 0}_{\alpha\rho k},
0,\ldots,0)\in\mathbb{R}^{N}, {\rm if} ~m\leq\rho k,
  \end{align*}
where $\|x_{0}\|_{1,w}=k$, $\|\eta_{0}\|_{1, w}\leq
m\cdot\frac{k}{m'}<k$. Obviously, $x_{0}$ is $k-$sparse,
$x_{1}=\frac{1}{\sqrt{k+\frac{mk^{2}}{m'^{2}}}}(x_{0}-\eta_{0})$ and
$\|\eta_{0}\|_{1, w}< \|x_{0}\|_{1,w}$. In view of $Ax_{1}=0$, we
have $Ax_{0}=A\eta_{0}$.

Thus, in the noiseless case $y=Ax_{0}$, suppose that the weighted
$l_{1}$ minimization method (\ref{f2}) can exactly recover $x_{0}$,
 i.e., $\widehat{x}=x_{0}$. According to the definition of $\widehat{x}$ and $y=A\eta_{0}$,
 it contradicts that $\|\eta_{0}\|_{1, w}< \|x_{0}\|_{1,w}=\|\widehat{x}\|_{1,w}$.

 In the noise case $y=Ax_{0}+z$, suppose that the weighted $l_{1}$ minimization method (\ref{f2}) can stably recover $x_{0}$,
 i.e., $\lim\limits_{z\rightarrow0}\widehat{x}=x_{0}$. We observe that $y-A(\widehat{x}-x_{0}+\eta_{0})=y-A\widehat{x}\in \mathcal{B}$,
thus $\|\widehat{x}\|_{1,w}\leq
\|\widehat{x}-x_{0}+\eta_{0}\|_{1,w}$. As $z\rightarrow 0$,
$\|x_{0}\|_{1,w}\leq \|\eta_{0}\|_{1,w}$. It contradicts that
$\|\eta_{0}\|_{1, w}< \|x_{0}\|_{1,w}$.

Hence, the weighted $l_{1}$ minimization method (\ref{f2}) fails to
exactly and stably recover $x_{0}$ based on $y$ and $A$.
\end{proof}



\begin{thebibliography}{12}
\bibitem{BS}
R. Baraniuk and P. Steeghs, Compressive radar imaging, in Proc. IEEE
Radar Conf., 2007, pp. 128-133.

\bibitem{BMP}
R. V. Borries,  C. Miosso and C. Potes, Compressed sensing using
prior information, in 2nd IEEE Int. Workshop on Computational
Advances in Multi-Sensor Adaptive Processing, CAMPSAP 2077, 12-14,
2007, pp. 121-124.

\bibitem{CWX1}
T. T. Cai, L. Wang and G. W. Xu, Shifting inequality and recovery of
sparse signals, IEEE Trans. Signal Process, 58(3), pp. 1300-1308,
2010.

\bibitem{CWX}
T. T. Cai, L. Wang and G. W. Xu, New bounds for restricted isometry
constants, IEEE Trans. Inform. Theory, 56(9), pp. 4388-4394, 2010.

\bibitem{CXZ}
T. T. Cai, G. W. Xu, J. Zhang, On recovery of sparse signal via
$l_{1}$ minimization, IEEE Trans. Inf. Theory, 55(7), pp. 3388-3397,
2009.

\bibitem{CZ}
T. T. Cai and A. Zhang, Spares representation of a polytope and
recovery of sparse signals and low-rank matrices, IEEE Trans.
Inform. Theory, 60(1), pp. 122-132, 2014.

\bibitem{CZ2}
T. T. Cai and A. Zhang, Compressed sensing and affine rank
minimization under restricted isometry, IEEE Trans. Signal Process.,
61(13), pp. 3279-3290, 2013.

\bibitem{CZ1}
T. T. Cai and A. Zhang, Sharp RIP bound for sparse signal and
low-rank matrix recovery, Appl. Comput. Harmon. Anal., 35, pp.
74-93, 2013.


\bibitem{CRT}
E. J. Cand\`{e}, J. Romberg and T. Tao, Stable signal recovery from
incomplete and inaccurate measurements, Commum. Pure Appl. Math.,
59, pp. 1207-1223, 2006.

\bibitem{CT}
E. Cand\`{e}s and T. Tao, Decoding by linear programming, IEEE
Trans. Inf. Theory, 51(12), pp. 4203-4215, 2005.

\bibitem{DF}
Q. Du and J. E. Fowler, Hyperspectral image compression using
jpeg2000 and principal component analysis, IEEE Geosci. Remote Sens.
Lett., 4(4), pp. 201-205, 2007.

\bibitem{FMSY}
M. P. Friedlander, H. Mansour, R. Saab and O. Yilmaz, Recoverying
compressively sampled signals using partial support information,
IEEE Transactions Information Theory, 58(2), pp. 1122-1134, 2012.

\bibitem{HS}
M. Herman and T. Strohmer, High-resoluton radar via compressed
sensing, IEEE Signal Process. Mag., 57(6), pp. 2275-2284, 2009.

\bibitem{J}
L. Jacques, A short note compressed sensing with partially known
signal support, Signal Process., 90, pp. 3308-3312, 2010.

\bibitem{KXAH}
M. A. Khajehnejad, W. Xu, A. S. Avestimehr and B. Hassibi, Weighted
$l_{1}$ minimization for sparse recovery with prior information, in
IEEE Int. Symp. Information Theory, ISIT 2009, 2009, pp. 483-487.

\bibitem{LV}
W. Lu and N. Vaswani, Exact reconstruction conditions and error
bounds for regularized modified basis pursuit, in Proc. Asilomar
Conf. on Signals, Systems and Computers, 2010.

\bibitem{LV1}
W. Lu and N. Vaswani, Modified basis pursuit denoising
(modifiedbpdn) for noisy compressive sensing with partially known
signal support, in IEEE Int. Conf. Acoustics Speech and Signal
Processing (ICASSP), 2010, 14-19, 2010, pp. 3926-3929.

\bibitem{LSLDP}
M. Lustig, J. M. Santos, J. H. Lee, D. L. Donoho and J. M. Pauly,
Application of compressed sensing for rapid MR inmaging, in SPARS,
(Rennes, France), 2005.

\bibitem{ML}
Q. Mo and S. Li, New bounds on the restricted isometry constant
$\delta_{2k}$, Appl. Comput. Harmon. Anal., 31(3), pp. 3335460-468,
2011.

\bibitem{PVMH}
F. Parvaresh, H. Vikalo, S. Misra and B. Hassibi, Recovering sparse
signals using sparse measurement matrices in compressed DNA
microarrays, IEEE J. Sel. Top. Signal Process., 2(3), pp. 275-285,
2008.

\bibitem{VL}
N. Vaswani and W. Lu, Modified-CS: Modifying compressive sensing for
problems with partially known support, IEEE Trans. Signal Process.,
58(9), pp. 4595-4607, 2010.

\bibitem{VL1}
N. Vaswani and W. Lu, Modified-CS: Modifying compressive sensing for
problems with partially known support, IEEE Int. Symp. Information
Theory, ISIT 2009, pp. 488-492, 2009.

\bibitem{ZZZ}
J. Zhang, D. Zhu and G. Zhang, Adaptive compressed sensing radar
oriented toward cognitive decetion in dynamic sparse target scene,
IEEE Trans. Signal Process., 60(4), pp. 1718-1729, 2012.










\end{thebibliography}
\end{document}